\documentclass[journal]{IEEEtran}

\usepackage{algorithm}
\usepackage{algorithmic}
\usepackage{amsmath}
\usepackage{amssymb}
\usepackage{latexsym}
\usepackage{multirow}
\usepackage{epsfig}
\usepackage{graphics}
\usepackage{graphicx}
\usepackage{mathrsfs}
\usepackage{subfigure}
\usepackage{slashbox}
\usepackage{mathrsfs}
\usepackage{pifont}
\usepackage{bbding}
\usepackage{stmaryrd}
\usepackage{amssymb}
\usepackage{amsfonts}
\usepackage{epic}
\usepackage{curves}
\usepackage{stfloats}
\usepackage{epstopdf}
\usepackage{cite}
\usepackage{stfloats}

\newcommand{\beq}{\begin{equation}}
\newcommand{\enq}{\end{equation}}
\newcommand{\ben}{\begin{eqnarray}}
\newcommand{\enn}{\end{eqnarray}}
\newcommand{\bei}{\begin{itemize}}
\newcommand{\eni}{\end{itemize}}

\newcommand{\bm}[1]{\mbox{\boldmath{$#1$}}}

\newtheorem{theorem}{Theorem}
\newtheorem{proposition}{Proposition}

\newtheorem{lemma}{Lemma}
\newtheorem{corollary}{Corollary}
\newtheorem{remark}{Remark}

\makeatletter
\newcommand{\figcaption}{\def\@captype{figure}\caption}
\newcommand{\tabcaption}{\def\@captype{table}\caption}
\makeatother
\date{}

\begin{titlepage}

\title{
\vspace{-1.0cm}
   \hfill{\em\small{}}\\
   \vspace{0.6cm} \LARGE
\begin{center}
{Achieving Global Optimality for Joint Source and Relay Beamforming Design in Two-Hop Relay Channels}\end{center}
}

\author{Hongying~Tang, Wen~Chen,~\IEEEmembership{Senior Member,~IEEE},\\ Jun Li, ~\IEEEmembership{Member,~IEEE}, Haibin Wan 
\thanks{Copyright (c) 2013 IEEE. Personal use of this material is permitted. However, permission to use this material for any other purposes must be obtained from the IEEE by sending a request to pubs-permissions@ieee.org.}
\thanks{Hongying~Tang, Wen~Chen are with the Department of Electronic Engineering,
Shanghai Jiaotong University, Shanghai, China, 200240, (e-mail: \{lojordan, wenchen\}@sjtu.edu.cn). Jun Li is with school of Eletrical and Information Engineering, University
of Sydney, Australia (e-mail: jun.li@sydney.edu.au). Haibin Wan is the  College of Computer Science and Electronic Information, Guangxi University, China (e-mail: hbwan@gxu.edu.cn).}
\thanks{This work is supported by the National 973 Project \#2012CB316106,
by NSF China \#61161130529 and \#61328101, by STCSM Science and
Technology Innovation Program \#13510711200, and by SEU National Key
Lab on Mobile Communications \#2013D11.}
}

\end{titlepage}

\begin{document}

\maketitle

\begin{abstract}
This paper deals with joint source and relay beamforming (BF) design for an amplify-and-forward (AF) multi-antenna multirelay network. Considering that the channel state information (CSI) from relays to destination is imperfect, we aim to maximize the worst case received signal-to-noise ratio (SNR). The associated optimization problem is then solved in two steps. In the first step, by fixing the source BF vector, a semi-closed form solution of the relay BF matrices is obtained, up to a power allocation factor. In the second step, the global optimal  source BF vector is obtained based on the  Polyblock outer Approximation (PA) algorithm.  We also propose two low-complexity methods for obtaining the source BF vector,  which are different in their complexities and performances. The optimal joint source-relay BF solution obtained by the proposed algorithms serves as the benchmark for evaluating the existing schemes and the proposed low-complexity methods. Simulation results show that the proposed robust design can significantly reduce the sensitivity of the channel uncertainty to the system performance.

\end{abstract}

\begin{IEEEkeywords}
Amplify-and-forward, multi-antenna multirelay system, global optimal, beamforming
\end{IEEEkeywords}
\IEEEpeerreviewmaketitle


\section{Introduction}\label{sec:1}
Relay communication can extend the coverage of wireless network and improve the spatial diversity of cooperative systems. There are several cooperative schemes being widely used, i.e., the Amplify-and-Forward (AF) scheme, the Decode-and-Forward scheme \cite{liu}, the Filter-and-Forward \cite{FF_Gershman,Twoway_FF_Schober,AF_FF_Schober} scheme etc. Among them, the AF scheme is the most simple scheme and has been efficiently used to exploit the benefit of relaying in the two-hop relay channels \cite{wang, zhang, Power_Jing, BK_grass, AF-BF, PU_Robust_ICC, ganzhengtsp, HS_Worst, ZW-wnt, HS_Worst_MR}, the multiple access relay channels \cite{wan}, and the two-way relay channels \cite{FF_Gershman,Twoway_FF_Schober,Tao_2, Tao_4, Tao, Tao_3, Tao_1, Tao_5}.

Performing transmit beamforming (BF) at source and relay can achieve higher data rate\cite{Tao_2}\cite{Tao_4}. In particular, AF-BF was considered in the following works \cite{Power_Jing, BK_grass, AF-BF, PU_Robust_ICC, ganzhengtsp, HS_Worst, ZW-wnt, HS_Worst_MR, wan, Tao_2, Tao_4, Tao, Tao_3, Tao_1,Tao_5}. By maximizing the received SNR, \cite{Power_Jing} gives the analytical solution of the beamforming design in a single source and multiple single-antenna relay network.  \cite{BK_grass} considers a multi-antenna source and single multi-antenna relay network, and gives closed-form solutions for both the source BF vector and the relay BF matrix. By relaxing the single-antenna source and single relay assumption, \cite{AF-BF} considers the more general case with a multi-antenna source and multiple multi-antenna relay network, and gives the closed-form of the relay BF matrices and a suboptimal solution for the source BF vector.

Those works are all based on the perfect channel state information (CSI) assumption. However, in a practical system, the perfect CSI is usually hard to obtain, thus reducing the efficiency of beamforming design. Therefore, robust design taking imperfect CSI into account has attracted much attention \cite{PU_Robust_ICC,ganzhengtsp, HS_Worst, ZW-wnt, HS_Worst_MR, Tao_2, Tao_4, Tao, Tao_3, Tao_1,Tao_5}. In \cite{PU_Robust_ICC}\cite{ganzhengtsp}, the authors consider a robust distributed beamforming design in a wireless relay network by minimizing the total relay transmit power and maximizing the received signal to noise ratio (SNR), respectively.  In the very recent work \cite{HS_Worst}, the authors obtain a closed-form solution for a single antenna source-destination pair and a multi-antenna relay network and discover that the robust design has the consistent form as the nonrobust design. For the more general work in \cite{HS_Worst_MR}, where the source and the destination are equipped with multiple antennas, the authors prove that the robust relay optimization leads to a channel-diagonalizing structure and a closed-form solution is proposed. Robust design in a two-way relay system are also studied in \cite{Tao, Tao_3, Tao_1,Tao_5}, on the maximization of SNR criteria, the MMSE criteria and the minimization of transmit power criteria, respectively.

In this paper, we consider the AF-relay networks with one multi-antenna source, multiple multi-antenna relays and a single antenna destination, and address the joint beamforming design of source and relays under imperfect CSI cases. Joint source and relay beamforming design has been fully investigated in the two-way relay model in both perfect and imperfect CSI cases \cite{Tao_3}\cite{Tao_1}. For the two-hop relay networks, however, this problem has not been well solved till now. Even in the perfect CSI case, \cite{AF-BF} only provides a suboptimal solution for the source BF vector. In the robust case,
\cite{ganzhengtsp} \cite{HS_Worst} discuss the situation when the source or the relays are equipped with a single antenna; \cite{ZW-wnt} investigates the robust relay precoders based on the MMSE receiver and the RZF precoding without taking into account the effect of source beamforming vector. Considering the fact that a practical network may involve a multi-antenna source and multiple relays, it is necessary to investigate the joint source and relay beamforming for these general networks.

The main contributions of this paper are as follows:
\begin{enumerate}
\item Considering imperfect CSI of the second hop at relays, for a given source BF vector, we derive a semi-closed form expression of the relay BF matrices, up to a scalar power allocation factor. Next, we obtain the power allocation factor  through iteration between a Dinkelbach-based approach and a second order cone programming (SOCP) problem.
\item To derive the optimal source BF vector, we transform the original problem into a monotonic problem, which allows us to apply the Polyblock outer Approximation (PA) algorithm to solve the problem. This PA-based algorithm mainly serves as a benchmark for the performance evaluation, both in the perfect CSI case and the robust case.
\item To further reduce the computational complexity, two low-complexity methods are proposed,  which are different in their complexities and performances. Simulation results show that the proposed robust design can significantly reduce the sensitivity of the channel uncertainty to the system performance.
\end{enumerate}

This paper is organized as follows: Section~\ref{sec:2} introduces the system model of the multi-antenna multi-relay channel and gives the problem formulation. In section~\ref{sec:opRBF}, we give  the semi-closed form for the relay BF design under a fixed source BF vector, up to a power allocation factor,  and then propose a Dinkelbach-based algorithm for determining the corresponding power allocation factor. In section~\ref{sec:opSBF}, the global optimal and subtoptimal source BF vectors are obtained. Finally section~\ref{sec:simu} provides numerical examples to validate the proposed algorithms.

In this paper, $[\cdot]^*$, $[\cdot]^T$ and $[\cdot]^H$ respectively denote the conjugate, transpose and conjugate transpose of a matrix or a vector. $\mathbb R^N$ and $\mathbb C^N$ respectively denote the $N$ dimensional real field and complex field. $\mathbf e_i$ denotes a zero vector except that the $i$th element is one, $\mathbf 0_N$ and $\mathbf{I}_N$ respectively denote the $N$-dimensional zero vector and the identity matrix.  We will use boldface lowercase letters to denote column vectors and boldface uppercase letters to denote matrices. $||\mathbf x||_2$ and $\|\mathbf x\|_1$  denote the Euclidean norm  and  the absolute sum of vector $\mathbf x$, respectively. Vec$(\mathbf X)$ stacks the columns of matrix $\mathbf X$ into a vector. $|\mathbf x|\triangleq[|x_1|, \cdots, |x_N|]^T$ and $|\mathbf x|^2\triangleq[|x_1|^2, \cdots, |x_N|^2]^T$.  The positive semidefinite matrix $\mathbf X$ is denoted by $\mathbf X\succeq 0$. For $\mathbf x=[x_1,\cdots,x_N]^T, \mathbf y=[y_1,\cdots,y_N]^T\in \mathbb R^N, \mathbf x\geq \mathbf y$ means $x_i\geq y_i$ for $i=1\dots N$. The tr$(\cdot$) is the trace of a matrix. diag$[x_1, \cdots, x_N]$ denotes a diagonal matrix with the diagonal entries $x_1, \cdots, x_N$.  $\mathbf v^\bot$ and $\mathbf v^\|$ respectively denote the unit vectors parallel and perpendicular to $\mathbf v$. ${\bm\upsilon}(\mathbf X)$ denotes the normalized principal eigenvector of $\mathbf X$.

\section{Problem Statement}\label{sec:2}
\subsection{System Model}
Consider  a two-hop AF multi-antenna multirelay network as shown in Fig.~1. The relays process the signals received from the source by using linear operations and forward the processed signals to the destination. We assume that the source and the relay $i$ have $N_T$ and $M_i$ antennas, for $1\leq i\leq R$, respectively, and the destination only has a single antenna. Note that the direct link between the source and destination is not taken into account due to large scale fading. The signal transmission is completed through two hops. In the first hop, the source transmits the $N_T-$dimensional vector
\begin{figure}[htp]
    \centering
    \includegraphics[width=3.5in]{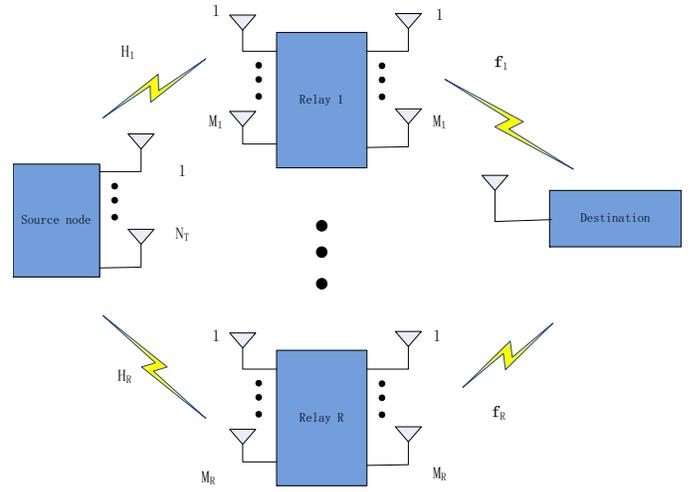}
    \caption{A two-hop multiple antenna multi-relay network.}\label{rank}
\end{figure}
\beq
\mathbf x=\mathbf gd,\nonumber
\enq
where $\mathbf g\in \mathbb C^{N_T}$ denotes the beamforming (BF)  vector at the source, and $d$ is the transmitted symbol with variance $\sigma_d^2=\mathcal E\{|d|^2\}=1$. The signal received by the relay $i$, $1\leq i\leq R$, is given by
\beq
\mathbf q_i=\mathbf H_i\mathbf x+\mathbf n_{i},\nonumber
\enq
where $\mathbf H_i\in \mathbb C^{M_i\times N_T}$ denotes the first hop channel from the source to the $i$th relay, and $\mathbf n_{i}\in \mathbb{C}^{M_i}$ denotes the additive white Gaussian noise (AWGN) vector with the covariance matrix $\sigma^2_R\mathbf{I}_{M_i}$ at relay $i$. By the AF strategy, the signal forwarded by relay $i$ is
\beq
\mathbf s_i=\mathbf B_i\mathbf q_i,\nonumber
\enq
where $\mathbf B_i\in \mathbb C^{M_i \times M_i}$ is the linear precoding matrix of relay $i$.
The received signal at the destination node can thus be expressed as
\ben
r&=&\sum_{i=1}^R\mathbf f_i^T\mathbf s_i+n_D\nonumber\\
&=&\sum_{i=1}^R\mathbf f_i^T\mathbf B_i\mathbf H_i\mathbf gd+\sum_{i=1}^R\mathbf f_i^T\mathbf B_i\mathbf n_{i}+n_D,\nonumber
\enn
where $\mathbf f_i$ denotes the channel from relay $i$ to the destination and $n_D$ is the additive white Gaussian noise (AWGN) observed at the destination with variance $\sigma^2_D$.

\subsection{Channel Uncertainty}\label{subsec:2b}
In a practical wireless communication scenario, perfect CSI is usually difficult to obtain. With only imperfect CSI, the system performance will be deteriorated. This motivates us to investigate the robust design taking the CSI errors into account. As will be verified in the simulations, our proposed robust scheme will significantly reduce the sensitivity of the system to uncertain CSI.

In this paper, we assume that the uncertainty of the first hop channel at the source is negligible and  model the CSI in the second hop at relays to be imperfect, more specifically,
\ben\label{equ:roba1}
\mathbf f_i=\tilde {\mathbf f}_i+\triangle\mathbf f_i,
\enn
where $\tilde {\mathbf  f}_i$ is the available CSI known at the $i$th relay, and $\triangle\mathbf f_i$ is the corresponding CSI error vector.  Under the circumstance when the source (e.g. a base station) and the relays are considered fixed, and the destination is moving (e.g. a mobile terminal), the channel statistics of the two hops are different. The first hop is undergoing a slow fading channel, whereas the second hop channel may be fast fading due to the mobility of the destination. Then the CSI feedback from the destination to the relays are usually outdated, and the channel uncertainty must be considered. \cite{CK_MulPoint} also uses this model for exploiting the situation when the relays are located closer to the source than to the destination, while this assumption is reasonable because of the high signal quality between the source and the relays.

Many existing works~\cite{ganzhengtsp,Tao,nonrob}  assume that the CSI error is bounded in a bundle manner, i.e, $\|\triangle \mathbf f\|_2\leq \varepsilon$ for some small $\varepsilon>0$, where $\triangle \mathbf f\triangleq [\triangle \mathbf f_1^T, \cdots, \triangle \mathbf f_R^T]^T$. However, this model is very conservative, as the channel between each relay node and the destination experiences independent distribution.  In this paper we adopt a more practical model, assuming that the CSI error vectors are estimated independently, i.e., $\|\triangle\mathbf f_i\|_2\leq \varepsilon_i$, for some small $\varepsilon_i>0$. We rewrite it as $\triangle\mathbf f \in \mathcal A$, where
\ben\label{equ:errorf}
\mathcal A\triangleq \{\mathbf a| \mathbf a=[\mathbf a_1^T, \cdots, \mathbf a_R^T]^T, \|\mathbf a_i\|_2\leq\varepsilon_i, \mathbf a_i\in \mathbb C^{M_i}\}\label{ellip}.
\enn
We also assume in this paper that the uncertainty error bound is not too large, i.e.,  $\varepsilon_i\leq\|\tilde {\mathbf f}_i\|_2 $,  which is reasonable since large error bound would lead to the instability of the system and any beamforming design becomes trivial. In this error model, one cannot use the S-lemma to transform the infinitely many constraints of the error vector into a linear matrix inequality (LMI) \cite{ganzhengtsp}\cite{Tao}, as it will degrade into a conservative approach \cite{lectures}.
By contrast, we will use an alternative approach based on the idea of \emph{real-valued implementation} proposed in \cite{LZ_Realvalue}, and prove in section \ref{sec:robb} that only finite realizations of the channel can act as the worst case channel, thus making the optimization problem tractable again.

\subsection{Problem Formulation}
By maximization of the worst case received SNR over the channel uncertainty region under individual power constraints at the relays and the source, the problem of jointly optimizing the source BF and the relay BFs can be mathematically formulated as
\begin{subequations}\label{equ:roba3}
\ben
\max_{\{\mathbf B_i\}_{i=1}^R, \mathbf g} \min_{\triangle \mathbf f\in \mathcal A}&\frac{|\sum_{i=1}^R\mathbf f_i^T\mathbf B_i\mathbf H_i\mathbf g|^2}{\sigma_D^2+\sigma_R^2\sum_{i=1}^R\|\mathbf f_i^T\mathbf B_i\|_2^2},\label{equ2}\\
\text{s.t}.& \|\mathbf B_i\mathbf H_i\mathbf g\|_2^2+\sigma_R^2\text{tr}(\mathbf{B}_i^H\mathbf{B}_i)\leq P_i,\forall i,\\
& \|\mathbf g\|_2^2\leq P_s.
\enn
\end{subequations}
where $P_s$ is the maximum power at the source and $P_i$ is the maximum power at relay $i$.
In  section \ref{sec:opRBF}, we first fix the source BF vector $\mathbf g$, and derive a semi-closed form of the optimal relay BF matrices up to a real-valued power allocation factor, which can be determined by an SOCP problem. Then in section \ref{sec:opSBF}, we propose a global optimal as well as two suboptimal algorithms to determine $\mathbf g$.

\section{Optimal BF matrices at relays}\label{sec:opRBF}
By fixing the source BF vector $\mathbf g$ and  taking into account the CSI error model \eqref{equ:roba1} and \eqref{equ:errorf}, problem \eqref{equ:roba3} becomes
\begin{subequations}\label{equ:robb1}
\ben
\max_{\{\mathbf B_i\}_{i=1}^R} \min_{\triangle \mathbf f\in \mathcal A} &\frac{|\sum_{i=1}^R (\tilde {\mathbf  f}_i+\triangle\mathbf f_i)^T\mathbf B_i\mathbf u_i|^2}{\sigma_D^2+\sigma_R^2\sum_{i=1}^R\| (\tilde {\mathbf f}_i+\triangle\mathbf f_i)^T\mathbf B_i\|_2^2},\label{equ2}\\
\text{s.t}.& \|\mathbf B_i\mathbf u_i\|_2^2+\sigma_R^2\text{tr}(\mathbf B_i^H\mathbf{B}_i)\leq P_i,\forall i,
\enn
\end{subequations}
where we defined $\mathbf u_i\triangleq \mathbf H_i\mathbf g$ for convenience. In section \ref{sec:opRBFa}, we will first introduce the related work of problem \eqref{equ:robb1}.  By fixing the source BF vector $\mathbf g$, a semi-closed form of $\mathbf B_i$ is given in section \ref{sec:robb}, up to a power allocation factor. Then in section \ref{subsec:opRBFb}, the optimal power allocation factor is determined via a Dinkelbach-based algorithm.
\subsection{Related Work}\label{sec:opRBFa}
Problem \eqref{equ:robb1} has been discussed in~\cite{nonrob}, where the authors consider the problem in the multipoint-to-multipoint setting. By vectorizing all $\mathbf B_i$ and stacking them to form a column vector as
\ben
\mathbf b_L\triangleq [\text{vec}(\mathbf B_1)^T, \cdots, \text{vec}(\mathbf B_R)^T]^T\in \mathbb C^{\sum_{i=1}^RM_i^2},\nonumber
\enn
then after some tedious manipulations, \eqref{equ:robb1} can be transformed into a semidefinite programming (SDP) problem with variable $\mathbf B\triangleq\mathbf b_L\mathbf b_L^H\in \mathbb C^{({\sum_{i=1}^RM_i^2})\times ({\sum_{i=1}^RM_i^2})}$. Obviously, this leads to prohibitively computational complexity. In addition, in some cases, the optimal $\mathbf B$ obtained by the SDP solver may not be of rank one, thus leading to suboptimal $\mathbf b_L$. Furthermore, the result in \cite{nonrob} is numerical and cannot provide any insight to the structure of the optimal relay BF matrices. Therefore it is necessary to re-investigate problem \eqref{equ:robb1}.

Recently, a closed form solution of \eqref{equ:robb1} when $R=1$ is derived in \cite{HS_Worst}.  Adopting the \emph{ saddle point} theorem, the authors prove that the worst-case CSI uncertainty  can be uniquely determined.  Additionally, the authors show that the robust relay BF matrix has a consistent form as that in the perfect CSI case.  However,  when the multiple relay channel is considered, the analysis becomes much more difficult and the extension of the saddle-point-based technique is not straightforward. In the next subsection, we will prove that the robust relay BF matrices in \eqref{equ:robb1} also have a similar form as that in the perfect CSI case, and the worst case CSI uncertainty is one of the $2^R$ possible channel errors (see Theorem~\ref{theorem:robb1}).

\subsection{The Semi-closed Form of Optimal Relay BF Matrices}\label{sec:robb}
We first introduce the following result given in \cite{Power_Jing}\cite{AF-BF} under perfect CSI assumption, based on which, we show the result of robust design.
\begin{lemma}[\cite{AF-BF}]\label{theorem:perfectCSI}
With perfect CSI assumption, i.e., $\Delta \mathbf f_i=\mathbf{0}_{M_i}$, the optimal relay BF matrices in \eqref{equ:robb1} are given by
\ben
\mathbf B_i=c^\sharp_i\mathbf {\hat f}_{i}^{\ast}\mathbf {\hat u}_{i}^H.\label{equ:nonrobg}
\enn
where $\mathbf {\hat u}_{i}\triangleq\mathbf u_i/\|\mathbf u_i\|_2$, and $\mathbf{ \hat f}_{i}\triangleq\mathbf f_i/\|\mathbf f_i\|_2$. The real valued power allocation vector $\mathbf c^\sharp\triangleq[c_1^\sharp, \cdots, c_R^\sharp]^T$ is determined by
\begin{subequations}\label{equ:nonc}
\ben
\mathbf{c}^\sharp=\arg\max_{\mathbf c=[c_1,\cdots,c_R]}& \frac{\bigg(\sum_{i=1}^Rc_i\|\mathbf f_i\|_2\|\mathbf u_i\|_2\bigg)^2}{\sigma_R^2\sum_{i=1}^Rc_i^2\|\mathbf f_i\|^2_2+\sigma_D^2}, \label{equ:7a}\\
\text{s.t.} &c_i\leq \sqrt{\frac{P_i}{\|\mathbf u_i\|_2^2+\sigma_R^2}},1\leq i\leq R\label{equ:7b}.
\enn
\end{subequations}
\end{lemma}
\begin{corollary}[\cite{Power_Jing}]\label{cor:Jing}
Define $\phi_i\triangleq\frac{\|\mathbf u_i\|_2\sqrt{1+\|\mathbf u_i\|_2^2}}{\|\mathbf f_i\|_2\sqrt{P_i}}$, for $i=1, \cdots, R$.  Let $\pi$ be a permutation of $\{1, \cdots, R\}$ such that $\{\phi_{\pi_{(i)}}\}_{i=1}^R$ are in descending order.
Then the $\mathbf c^\sharp$ in \eqref{equ:nonc} has the following analytical solution
\ben\nonumber
c_i^\sharp=\upsilon_i^{(j_0)}\sqrt{\frac{P_i}{\|\mathbf u_i\|_2^2+\sigma_R^2}},
\enn
where
\ben\nonumber
\upsilon_i^{(j)}\triangleq\left\{
\begin{array}{lcl}
1,&&i=\pi_1, \cdots, \pi_j,\\
\lambda_j\phi_i,&& i=\pi_{j+1}, \cdots, \pi_R,
\end{array}
\right.
\enn
$\lambda_j\triangleq\frac{1+\sum_{m=1}^ja^2_{\pi_m}}{\sum_{m=1}^jb_{\pi_m}}$, $a_j\triangleq\frac{\|\mathbf f_j\|_2\sqrt{P_j}}{\sqrt{1+\|\mathbf u_j\|_2^2}}$, $b_j\triangleq\frac{\|\mathbf f_j\|_2\|\mathbf u_j\|_2\sqrt{P_j}}{\sqrt{1+\|\mathbf u_j\|_2^2}}$,  and $j_0$ is the smallest $j$ such that $\lambda_j<\phi^{-1}_{\pi_{j+1}}$ for $1\leq j\leq R.$
\end{corollary}

Define $\mathcal B\triangleq \{\mathbf a|\mathbf a=[a_1, \cdots, a_R]^T, a_i=\|\tilde {\mathbf f}_i\|_2\pm \varepsilon_i\}$, and $\mathbf f_{\eta}\triangleq[f_{\eta 1}, \cdots, f_{\eta R}]^T$. Now we present the optimal robust relay BF matrices in Theorem~\ref{theorem:robb1}.
\begin{theorem}\label{theorem:robb1} The optimal robust relay BF matrices in \eqref{equ:robb1} are given by
\ben\label{equ:robb2}
\mathbf B_i=c^\sharp_i\hat{\tilde {\mathbf f}}_{i}^{\ast}\mathbf{\hat  u}_i^H.
\enn
where $\hat{\tilde {\mathbf f}}_{i} \triangleq\tilde { \mathbf  f}_{i}/\|\tilde { \mathbf f}_i\|_2$ and the real valued $\mathbf c^\sharp$ is the optimal solution to the following problem
\begin{subequations}\label{equ:robb3}
\ben
\max_{\mathbf c}\min_{\mathbf f_{\eta}\in \mathcal B}&& \frac{\left(\sum_{i=1}^Rf_{\eta i}c_i\|\mathbf u_i\|_2\right)^2}{\sigma_R^2\sum_{i=1}^Rf_{\eta i}^2c_i^2+\sigma_D^2},\\
\text{s.t.}&&c_i\leq \sqrt{\frac{P_i}{\sigma_R^2+\|\mathbf u_i\|_2^2}}, 1\leq i\leq R.
\enn
\end{subequations}
\end{theorem}

Notice that in section \ref{subsec:2b}, we have assumed that $\varepsilon_i\leq \|\tilde  {\mathbf f}_i\|_2$. Thus any vector $\mathbf f_{\eta}\in \mathcal B$ has nonnegative real valued elements. From Theorem \ref{theorem:robb1}, one can observe that problem $\eqref{equ:robb3}$ is only optimized over the discrete set $\mathcal B$ with $2^R$ elements. By contrast, the original problem \eqref{equ:robb1} is optimized over the continuous region $\mathcal A$ with infinite channel realizations. This important step significantly reduces the computational complexity and makes problem \eqref{equ:robb1} in a more tractable form.

To prove Theorem \ref{theorem:robb1}, we first discuss the structure of the optimal $\mathbf B_i$, whose expression is given in the following lemma.
\begin{lemma}\label{lemma:robb1}
The optimal $\mathbf B_i$ in \eqref{equ:robb1} must have the form $\mathbf B_i=\mathbf b_i \mathbf{\hat  u}_i^H$ for some $\mathbf b_i\in \mathbb C^R$. Denote $\mathbf f\triangleq[\mathbf f_1^T, \cdots, \mathbf f_R^T]^T=[(\tilde  {\mathbf f}_1+\triangle\mathbf f_1)^T, \cdots, (\tilde {\mathbf f}_R+\triangle\mathbf f_R)^T]^T$. Then \eqref{equ:robb1} becomes
\begin{subequations}\label{equ:robb4}
\ben
\max_{\mathbf b_i}\min_{\triangle \mathbf f\in \mathcal A}&\text{SNR}(\mathbf b_i, \mathbf f)\triangleq\frac{\left|\sum_{i=1}^R(\tilde {\mathbf f}_i+\triangle\mathbf f_i)^T\mathbf b_i\|\mathbf u_i\|_2\right|^2}{\sigma_R^2\sum_{i=1}^R\|(\tilde {\mathbf f}_i+\triangle\mathbf f_i)^T\mathbf b_i\|_2^2+\sigma_D^2}\\
\text{s.t.}&\|\mathbf b_i\|_2\leq \sqrt{\frac{P_i}{\sigma_R^2+\|\mathbf u_i\|_2^2}}.
\enn
\end{subequations}
\end{lemma}
\begin{proof}
See Appendix \ref{app:lemmarobb1}.
\end{proof}

Then we come to determine the optimal $\mathbf b_i$. To proceed, we first discuss a particular case $\mathbf b_i=c_i\hat{\tilde {\mathbf f}}_{i}^{\ast}$ for some $c_i\in \mathbb C$ as in the following lemma.

\begin{lemma}\label{lemma:robb2}
If $\mathbf b_i=c_i\hat{\tilde {\mathbf f}}_{i}^{\ast}$ for some $c_i\in \mathbb C$,
then the optimal $c_i$ of problem \eqref{equ:robb4} must be real-valued and problem \eqref{equ:robb4} can be transformed into
\begin{subequations}\label{equ:robb6}
\ben
\max_{\mathbf c}\min_{\mathbf f_{\eta}\in \mathcal B}&&\frac{\left(\sum_{i=1}^Rf_{\eta i}c_i\|\mathbf u_i\|_2\right)^2}{\sigma_R^2\sum_{i=1}^Rf_{\eta i}^2c_i^2+\sigma_D^2},\\
\text{s.t.}&&c_i\leq \sqrt{\frac{P_i}{\sigma_R^2+\|\mathbf u_i\|_2^2}},\,\, 1\leq i\leq R.\label{powercons}
\enn
\end{subequations}
\end{lemma}
\begin{proof}
See appendix \ref{app:lemmarobb2}.
\end{proof}

\begin{proof}[Proof of Theorem \ref{theorem:robb1}]
Denote the optimal solution of \eqref{equ:robb6} as $\mathbf c^{\sharp}\triangleq [c_1^{\sharp}, \cdots, c_R^{\sharp}]$ and $\mathbf f_{\eta}^{\sharp}\triangleq[f_{\eta 1}^{\sharp}, \cdots, f_{\eta R}^{\sharp}]^T$. In appendix \ref{app:lemmarobb2}, we have shown that when $\mathbf b_i=c_i^{\sharp}\hat{\tilde {\mathbf f}}_{i}^\ast$,  the corresponding \emph{worst channel }is $\mathbf f^{\sharp}\triangleq[f_{\eta 1}^{\sharp}\hat{\tilde {\mathbf f}}_{1}^T, \cdots, f_{\eta R}^{\sharp}\hat{\tilde {\mathbf f}}_{R}^T]^T$. When we use the term \emph{worst channel}, we mean the channel $\mathbf f$ with the minimum SNR over $\mathbf f_{\eta}\in \mathcal B$ under a fixed $\mathbf c$ in \eqref{equ:robb6}.
Hence we have
\ben\label{equ:robb7}
\min _{\mathbf f_{\eta}\in \mathcal B}\text{SNR}(c_i^{\sharp}\hat{\tilde {\mathbf f}}_{i}^\ast,  \mathbf f)=\text{SNR}(c_i^{\sharp}\hat{\tilde {\mathbf f}}_{i}^\ast, \mathbf f^{\sharp}).
\enn

Consider the received SNR in \eqref{equ:robb4} with any $\mathbf b_i$ under the particular channel $\mathbf f^{\sharp}$.  We can decompose $\mathbf b_i\in \mathbb C^{M_i}$ as $\mathbf b_i=c_i(\tilde{ \mathbf f}_i^{\|})^{\ast}+d_i(\tilde { \mathbf  f}_i^{\bot})^{\ast}$, where $c_i, d_i\in \mathbb C$, and $\sqrt{|c_i|^2+|d_i|^2}=\|\mathbf b_i
\|_2$. Then  we have
\ben
&&\text{SNR}(\mathbf b_i, \mathbf f^{\sharp})\nonumber\\
&=&\frac{\bigg|\sum_{i=1}^Rf_{\eta i}^{\sharp}\hat{\tilde {\mathbf f}}_i^T(c_i(\tilde{\mathbf f}_i^{\|})^{\ast}+d_i(\tilde { \mathbf  f}_i^{\bot})^{\ast})\|\mathbf u_i\|_2\bigg|^2}{\sigma_R^2\sum_{i=1}^R\left\|f_{\eta i}^{\sharp}\hat{\tilde {\mathbf f}}_{i}^T(c_i(\tilde{ \mathbf  f}_i^{\|})^{\ast}+d_i(\tilde{ \mathbf f}_i^{\bot})^{\ast})\right\|_2^2+\sigma_D^2},\nonumber\\
&=&\frac{\bigg|\sum_{i=1}^Rf_{\eta i}^{\sharp}c_i\|\mathbf u_i\|_2\bigg|^2}{\sigma_R^2\sum_{i=1}^Rf_{\eta i}^{\sharp 2}|c_i|^2+\sigma_D^2}, \nonumber\\
&\overset{(a)}\leq &\frac{\bigg|\sum_{i=1}^Rf_{\eta i}^{\sharp}c_i^{\sharp}\|\mathbf u_i\|_2\bigg|^2}{\sigma_R^2\sum_{i=1}^Rf_{\eta i}^{\sharp 2}|c_i^{\sharp}|^2+\sigma_D^2},\nonumber\\
&=&\text{SNR}(c_i^{\sharp}\hat{\tilde {\mathbf f}}_{i}^\ast, \mathbf f^{\sharp}),\label{equ:s4}
\enn
where ($a$) is due to the fact that the optimal solution in \eqref{equ:robb6} is $\mathbf c^{\sharp}$.
Since $\mathbf f^{\sharp}$ is only a particular channel, there must be
\ben\label{equ:robb8}
\min _{\mathbf f_{\eta}\in \mathcal B}\text{SNR}(\mathbf b_i, \mathbf f) \leq\text{SNR}(\mathbf b_i,  \mathbf f^{\sharp}).
\enn
Combing \eqref{equ:robb7} \eqref{equ:s4}, and \eqref{equ:robb8}, we have
\begin{multline}\label{equ:ieqb}
  \min _{\mathbf f_{\eta}\in \mathcal B}\text{SNR}(\mathbf b_i, \mathbf f) \leq\text{SNR}(\mathbf b_i,  \mathbf f^{\sharp})\\
  \overset{(a)}\leq \text{SNR}(c_i^{\sharp}\hat{\tilde {\mathbf f}}_{i}^\ast, \mathbf f^{\sharp})\overset{(b)}=\min _{\mathbf f_{\eta}\in \mathcal B}\text{SNR}(c_i^{\sharp}\hat{\tilde {\mathbf f}}_{i}^\ast ,\mathbf f),
\end{multline}
where $(a)$ is due to \eqref{equ:s4} and $(b)$ is due to \eqref{equ:robb7}. \eqref{equ:ieqb} shows that the optimal $\mathbf b_i^{\sharp}=c_i^{\sharp}\hat{\tilde {\mathbf f}}_{i}^\ast$. By the above discussion, combining Lemma \ref{lemma:robb1} and Lemma \ref{lemma:robb2}, we get the semi-closed form of $\mathbf B_i$ as in \eqref{equ:robb2}, up to a power allocation factor $\mathbf c$ determined by \eqref{equ:robb3}.
\end{proof}

\subsection{Dinkelbach based Algorithm for Solving the Optimal Power Allocation Factor $\mathbf c$}\label{subsec:opRBFb}
\newcounter{MYtempeqncnt}
\begin{figure*}[b!]
\hrulefill \vspace*{4pt}
\normalsize
\setcounter{MYtempeqncnt}{\value{equation}}
\setcounter{equation}{15}
\begin{subequations}\label{equ:f2}
\ben
\max_{\mathbf c}\min_{\mathbf f_{\eta}\in \mathcal B}&& \sum_{i=1}^Rf_{\eta i}c_i\|\mathbf u_i\|_2-\sqrt{\gamma^{(\kappa)}}\sigma_R^2\sum_{i=1}^Rf_{\eta i}^2c_i^2-\sqrt{\gamma^{(\kappa)}}\sigma_D^2,\\
\text{s.t.}&&c_i\leq \sqrt{\frac{P_i}{\sigma_R^2+\|\mathbf u_i\|_2^2}}.
\enn
\end{subequations}
\setcounter{equation}{\value{MYtempeqncnt}}
\end{figure*}
\begin{figure*}[b!]
\hrulefill \vspace*{4pt}
\normalsize
\setcounter{MYtempeqncnt}{\value{equation}}
\setcounter{equation}{16}
\begin{subequations}\label{equ:f3}
\ben
\max_{\mathbf c, \tau}&& \tau,\\
\text{s.t.}&&\min_{\mathbf f_{\eta}\in \mathcal B}\sum_{i=1}^Rf_{\eta i}c_i\|\mathbf u_i\|_2-\sqrt{\gamma^{(\kappa)}}\sigma_R^2\sum_{i=1}^Rf_{\eta i}^2c_i^2-\sqrt{\gamma^{(\kappa)}}\sigma_D^2\geq \tau,\\
&&c_i\leq \sqrt{\frac{P_i}{\sigma_R^2+\|\mathbf u_i\|_2^2}}.
\enn
\end{subequations}
\setcounter{equation}{\value{MYtempeqncnt}}
\end{figure*}
\begin{figure*}[b!]
\hrulefill \vspace*{4pt}
\normalsize
\setcounter{equation}{17}
\begin{subequations}\label{equ:f4}
\ben
\min_{\mathbf c, \tau}&& -\tau,\\
\text{s.t.}&&\sum_{i=1}^Rf_{\eta i}c_i\|\mathbf u_i\|_2-\sqrt{\gamma^{(\kappa)}}\sigma_R^2\sum_{i=1}^Rf_{\eta i}^2c_i^2-\sqrt{\gamma^{(\kappa)}}\sigma_D^2\geq \tau, \mathbf f_{\eta}\in \mathcal B,\\
&&c_i\leq \sqrt{\frac{P_i}{\sigma_R^2+\|\mathbf u_i\|_2^2}}.
\enn
\end{subequations}
\setcounter{equation}{\value{MYtempeqncnt}}
\end{figure*}
In Lemma~\ref{theorem:perfectCSI} under perfect CSI assumption, $\mathbf c$
is obtained by a closed form solution in Corollary~\ref{cor:Jing}. However, in the robust case, such explicit analytical result is difficult to be derived. In this subsection, we will present a Dinkelbach based algorithm for solving $\mathbf c$ in~\eqref{equ:robb3}.

Introducing a slack variable $\gamma$, problem \eqref{equ:robb3} can be transferred into the following equivalent problem.
\begin{subequations}\label{equ:boxwenchen}
\ben
\max _{\mathbf c, \gamma}&& \gamma\\
\text{s.t.} &&\frac{\left(\sum_{i=1}^Rf_{\eta i}c_i\|\mathbf u_i\|_2\right)^2}{\sigma_R^2\sum_{i=1}^Rf_{\eta i}^2c_i^2+\sigma_D^2}\geq \gamma, \mathbf f_{\eta}\in \mathcal B,\\
&& c_i\leq \sqrt{\frac{P_i}{\sigma_R^2+\|\mathbf u_i\|_2^2}},
\enn
\end{subequations}
which can be solved by checking feasibility for a fixed $\gamma$ iteratively. To find the maximum value of $\gamma$, the conventional method is to use bisection approach~\cite{convex}: In the $\kappa$th iteration, assume that the optimal value $\gamma$ lies in the interval $[\gamma_l^{(\kappa)}, \gamma_u^{(\kappa)}]$. Set $\gamma=(\gamma_l^{(\kappa)}+\gamma_u^{(\kappa)})/2$ and solve \eqref{equ:boxwenchen}. If this problem is found to be feasible, update the interval bounds as $\gamma_l^{(\kappa+1)}=\gamma$ and $\gamma_u^{(\kappa+1)}=\gamma_u^{(\kappa)}$; Otherwise, update the interval bounds as $\gamma_l^{(\kappa+1)}=\gamma_l^{(\kappa)}$ and $\gamma_u^{(\kappa+1)}=\gamma$. This iteration is repeated until some threshold is achieved.

Since \eqref{equ:robb3} is a generalized fractional programming problem, it can be alternatively solved with the Dinkelbach-based algorithm as in \cite{WW_Dinkelbach} and \cite{ZF_Dinkelbach}. Unlike the bisection-based algorithm, the Dinkelbach-based algorithm does not need to shrink the interval iteratively. By contrast, it exploits the inherent property of the factional programming problem and approaches to the optimal $\gamma$ from the left side, e.g., $\gamma^{(\kappa)}\leq \gamma$.
The advantage of the Dinkelbach-based algorithm lies in the fact that it has a quotient-superlinear convergence, which is obviously faster than the linear convergence of bisection-based algorithm \cite{ZF_Dinkelbach}. Basically, the Dinkelbach-based algorithm aims to solve a sequence of problem \eqref{equ:f2} at the $\kappa$th iteration, as shown in the bottom of this page.

By introducing a slack variable $\tau$, problem \eqref{equ:f2} becomes \eqref{equ:f3},
which is obviously equivalent to
a second order cone programming (SOCP) problem \eqref{equ:f4}, and can be solved in polynomial time by interior point method. Then the solution $\mathbf c^{(\kappa)}$ from \eqref{equ:f4} is used to update $\gamma^{(\kappa+1)}$, i.e.,
\addtocounter{equation}{3}
\ben\label{equ:ier}
\gamma^{(\kappa+1)}=\min_{\mathbf f_{\eta}\in \mathcal B} \frac{\left(\sum_{i=1}^Rf_{\eta i}c_i^{(\kappa)}\|\mathbf u_i\|_2\right)^2}{\sigma_R^2\sum_{i=1}^Rf_{\eta i}^2c_i^{(\kappa)2}+\sigma_D^2}.
\enn
When $\tau=0$, this iteration stops.
We summarize the Dinkelbach based algorithm in the Algorithm~\ref{alg:dinkelbach}.
\begin{table}[!h]
\centering \caption{Algorithm I: Dinkelbach-based Algorithm for determining the optimal $\mathbf c$ in \eqref{equ:robb3}}
\begin{tabular}{|p{0.5cm}|p{7cm}|cl}\hline
1 & Choose $\delta_1$ as the desired threshold. Set $\kappa=0$ and  $c_i^{(\kappa)}=\sqrt{\frac{P_i}{\sigma_R^2+\|\mathbf u_i\|_2^2}}$ as the initial power allocation factor.\\\hline
2 &With given $\mathbf c^{(\kappa)}$, set $\gamma^{(\kappa+1)}$ as in \eqref{equ:ier}.
\\\hline
3 &Solve the SOCP problem in  \eqref{equ:f4} to obtain $\mathbf c^{(\kappa+1)}$.\\\hline
4 &If $\min_{\mathbf f_{\eta}\in \mathcal B}\sum_{i=1}^Rf_{\eta i}c_i^{(\kappa)}\|\mathbf u_i\|_2-\sqrt{\gamma^{(\kappa)}}\sigma_R^2\sum_{i=1}^R
f_{\eta i}^2c_i^{(\kappa)2}-\sqrt{\gamma^{(\kappa)}}\sigma_D^2
 \leq \delta_1$, go to Step $5$. Otherwise, $\kappa=\kappa+1$, and go to Step $2$. \\\hline
5&Return $\mathbf c^\sharp=\mathbf c^{(\kappa)}$.\\\hline
\end{tabular}\label{alg:dinkelbach}
\end{table}


\begin{remark}\label{full_power}
In the perfect CSI case, $c_i$ is obtained by Corollary~\ref{cor:Jing}, which can be any value between $0$ and its maximal value. However, as pointed out in \cite{Power_Jing}, there is at least one relay that uses its full power. The same phonomania holds true in the robust case. This can be explained as follows. Suppose that none of the relays uses its full power. Then, there exists a real-valued $\chi>1$ defined as
\ben
\chi\triangleq \min _{i\in \{1, \cdots, R\}}\Bigg\{\sqrt{\frac{P_i}{c_i^{\sharp 2}(\|\mathbf u_i\|_2^2+\sigma_R^2)}}\Bigg\}.\nonumber
\enn
It is easy to see that $\chi c_i^{\sharp}$ also satisfies the power constraints in \eqref{powercons}. But
\ben
\min_{\mathbf f_{\eta}\in \mathcal B}\frac{\chi^2\big(\sum_{i=1}^Rf_{\eta i}c_i^{\sharp}\|\mathbf u_i\|_2\big)^2}{\chi^2\sigma_R^2\sum_{i=1}^Rf_{\eta i}^2c_i^{\sharp 2}+\sigma_D^2}>\min_{\mathbf f_{\eta}\in \mathcal B}\frac{\big(\sum_{i=1}^Rf_{\eta i}c_i^{\sharp}\|\mathbf u_i\|_2\big)^2}{\sigma_R^2\sum_{i=1}^Rf_{\eta i}^2c_i^{\sharp 2}+\sigma_D^2}.\nonumber
\enn
Then the new coefficient $\chi c_i^{\sharp}$ leads to a higher SNR which contradicts to the assumption that  $c_i^{\sharp}$ is the optimal solution.
\end{remark}

\section{Optimal BF vector at the source}\label{sec:opSBF}
By Theorem~\ref{theorem:robb1}, the optimization variables of problem \eqref{equ:robb1} has been transformed into $\mathbf c$ and $\mathbf g$. According to section~\ref{sec:robb} and section~\ref{subsec:opRBFb}, by fixing a $\mathbf g$, the optimal solution of $\mathbf c$, can be obtained from Algorithm \ref{alg:dinkelbach}, i.e., $\mathbf c^\sharp=\mathbf c^\sharp(\mathbf g)$.
Then the remaining challenge is to determine the optimal $\mathbf g$, which is the solution of
\begin{subequations}\label{equ:boxnewu}
\ben
\max _{\mathbf g}\min_{\mathbf f _{\eta}\in \mathcal B}&&\frac{\big|\sum_{i=1}^Rf_{\eta i}c^{\sharp}_i(\mathbf g)\|\mathbf u_i\|_2\big|^2}{\sigma_R^2\sum_{i=1}^Rf_{\eta i}^2c^{\sharp}_i(\mathbf g)^2+\sigma_D^2},\\
\text{s.t.} &&\|\mathbf g\|_2^2\leq P_s.
\enn
\end{subequations}

Due to the non-convex nature of \eqref{equ:boxnewu}, it seems impossible to derive the optimal solution.  Even in the perfect CSI case, the authors in \cite{AF-BF} only propose a suboptimal algorithm based on the Gradient method. However, by exploiting the hidden monotonic property of problem \eqref{equ:boxnewu}, we propose an efficient algorithm based on the Polyblock outer Approximation (PA) algorithm to determine the global optimal $\mathbf g$. We also find that the global optimal $\mathbf g$ is parallel to the principal eigenvector of $\sum_{i=1}^R\mu_i\mathbf H_i^H\mathbf H_i$, for some $\sum_{i=1}^R\mu_i\leq 1, \mu_i\geq 0$ in section \ref{sec:opSBFa}.
Our result covers the special case discussed in \cite{BK_grass} that $\mathbf g=\sqrt{P_s}{\bm \upsilon}(\mathbf H_1^H\mathbf H_1)$ when $R=1$.
\subsection{Monotonic Optimization}\label{sec:opSBFa}
Let $\mathbb R_{+}^{N}$ be the $N$-dimensional non-negative real set. A set $\mathcal H\subset \mathbb R_{+}^{N}$ is called normal if for any point $\mathbf x\in \mathcal H$, any point $\mathbf x'$ with $\mathbf 0\leq \mathbf x'\leq \mathbf x$ must satisfy $\mathbf x'\in \mathcal H$. An optimization problem is the \emph{monotonic optimization} problem if it can be expressed as
\ben\nonumber
\max _{\mathbf x}\Phi(\mathbf x),\quad \quad \text{s.t.} ~\mathbf x\in \mathcal H,
\enn
where $\mathcal H\subset \mathbb R_{+}^{N}$ is a nonempty normal closed set and the function $\Phi(\mathbf x)$ is an increasing function with respect to $\mathbf x\in \mathcal H$.

To exploit the monotonic property of problem \eqref{equ:boxnewu}, we define
\ben
\mathbf w\triangleq[w_1, \cdots, w_R]^T\triangleq[\|\mathbf H_1\mathbf g\|_2^2, \cdots, \|\mathbf H_R\mathbf g\|_2^2]^T,\nonumber
\enn
Then the worst case SNR becomes a function of the new variable $\mathbf w$, i.e.,
\ben\label{equ:worstbu}
{\bf SNR}(\mathbf w) &\triangleq&\min_{\mathbf f_{\eta}\in \mathcal B}\frac{\big|\sum_{i=1}^Rf_{\eta i}c^{\sharp}_i(\mathbf w)\sqrt{w_i}\big|^2}{\sigma_R^2\sum_{i=1}^Rf_{\eta i}^2c^{\sharp}_i(\mathbf w)^2+\sigma_D^2},
\enn
where  $c_i^{\sharp}(\mathbf w)$ is the optimal solution of the following problem for given $\mathbf w$,
\begin{subequations}\label{equ:robbw}
\ben
\max_{\mathbf c}\min_{\mathbf f_{\eta}\in \mathcal B}&& \frac{\left(\sum_{i=1}^Rf_{\eta i}c_i\sqrt{w_i}\right)^2}{\sigma_R^2\sum_{i=1}^Rf_{\eta i}^2c_i^2+\sigma_D^2},\\
\text{s.t.}&&c_i\leq \sqrt{\frac{P_i}{\sigma_R^2+w_i}}.
\enn
\end{subequations}
Denote
\ben\label{equ:regionu2}
\mathcal U&\triangleq &\{\mathbf w|\mathbf w=[\text{tr}(\mathbf H^T_1\mathbf H_1\mathbf G),\cdots, \nonumber\\
&&\text{tr}(\mathbf H^T_R\mathbf H_R\mathbf G)]^T,\, \mathbf G\succeq \mathbf 0,\, \text{tr}(\mathbf G)\leq P_s\}.\nonumber
\enn
Then we have the following proposition.
\begin{proposition}\label{pro:monotonic}
Problem \eqref{equ:boxnewu} is equivalent to the following \emph{monotonic optimization} problem
\ben\label{equ:pa2}
\max_{\mathbf w}{\bf SNR}(\mathbf w), \quad \quad \text{s.t.}~ \mathbf w\in \mathcal U,
\enn
where the optimal $\mathbf w$ of \eqref{equ:pa2} must be on the Pareto boundary\footnote{$\mathbf x$ is called the Pareto boundary (or Pareto optimal) of a region $\mathcal H$ if there is no other vector $\mathbf x'\in \mathcal H$ such that $\mathbf x'> \mathbf x$.} of $\mathcal U$, and the associated $\mathbf G$ must be of rank one.
\end{proposition}
\begin{proof}
See Appendix \ref{app:monotonic}.
\end{proof}

Suppose that $\mathbf G^\sharp$ is associated with the optimal $\mathbf w$ of \eqref{equ:pa2}. According to Proposition~\ref{pro:monotonic}, $\mathbf G^\sharp$ must be of rank one. By eigenvalue decomposition $\mathbf G^\sharp=\mathbf g^\sharp\mathbf g^{\sharp H}$, we can obtain $\mathbf g^{\sharp}$.
Specifically, the structure of the global optimum $\mathbf g^\sharp$ can be derived in Corollary \ref{cor:pa} by following a similar argument as that in \cite{RM_Pareto}.

\begin{corollary}\label{cor:pa}
The global optimal $\mathbf g$ has the structure
\ben
\mathbf g=\sqrt{P_s}{\bm\upsilon}\left(\sum_{i=1}^R\mu_i\mathbf H_i^H\mathbf H_i\right), {\bm \mu}\triangleq[\mu_1, \cdots, \mu_R]\in \mathcal V.\label{equ:closeg}\nonumber
\enn
where
\ben
\mathcal V\triangleq\left\{{\bm \mu}\Big|\sum_{i=1}^R\mu_i=1, \mu_i\geq 0\right\}.\nonumber
\enn
\end{corollary}

\begin{remark}
According to Corollary \ref{cor:pa}, by implementing the grid search in $\mathcal V$, one can asymptotically achieve the optimal SNR if the grid is sufficiently fine. Set the search step as $0.01$, one has to compare $100$ points for $R=2$. When $R=3, 4, 5$, this number rises to $5000, 250000, 12500000$, respectively. It can be seen that the complexity of this grid search increases with $R$ rapidly. Hence in the next subsection, we will propose an efficient PA-based algorithm for solving the optimal $\mathbf w$ by taking advantage of the monotonic property of \eqref{equ:pa2}.
\end{remark}

\begin{remark}
It is worth pointing out that for some special cases, the optimal source BF vector $\mathbf g$ has following expressions
\begin{itemize}
  \item case $1$: $N_T=1$, then $g=\sqrt{P_s}$.
  \item case $2$: $R=1$,  then $\mathbf g=\sqrt{P_s}{\bm\upsilon}(\mathbf H_1^H\mathbf H_1).$
  \item case $3$: $M_1=M_2=1$, then
$\mathbf g=\sqrt{P_S}\sin\theta\frac{\Pi_{\mathbf h_2}\mathbf h_1}{\|\Pi_{\mathbf h_2}\mathbf h_1\|_2}+\sqrt{P_S}\cos\theta\frac{\Pi_{\mathbf h_2}^{\bot}\mathbf h_1}{\|\Pi_{\mathbf h_2}^{\bot}\mathbf h_1\|_2}$, where $
\theta\in [0, \frac{\pi}{2}]$, and can be obtained by one dimensional search, $\Pi_{\mathbf x}\triangleq \mathbf x(\mathbf x^H\mathbf x)^{-1}\mathbf x^H$ is the orthogonal projection onto the column space of $\mathbf x$, and $\Pi^{\bot}_{\mathbf x}\triangleq \mathbf I-\Pi_{\mathbf x}$ is the orthogonal projection onto the orthogonal complement of the column space of $\mathbf x$.
\end{itemize}
\end{remark}

\subsection{Polyblock outer Approximation (PA) Algorithm}\label{sec:robc}
In the literature, two general algorithms are widely used for solving monotonic problems: the PA algorithm from \cite{LL_ploy} and the Branch-Reduce-and-Bound (BRB) algorithm from \cite{EB_BRB}\cite{EB_BRB2}. In this subsection, we will briefly introduce the PA algorithm, and then propose a PA-based algorithm for  solve the optimal $\mathbf w$ in \eqref{equ:pa2}, which automatically results in the solution of global optimal source BF vector $\mathbf g$. Performance comparison between the PA and BRB algorithm will be given in our simulation part.
More details on PA algorithm can be found in~\cite{LL_ploy, EB_BRB2}.

A set $\mathcal P$ is called a polyblock if it is the union of a finite number of boxes\footnote{
For given $\mathbf b\in \mathbb R_{+}^{N_T}$, the set of all $\mathbf x$ such that $\mathbf 0\leq \mathbf x \leq \mathbf b$ is called a box with vertex $\mathbf b$.}. The main idea of PA is to approximate $\mathcal U$  by constructing a sequence of polyblocks $\mathcal P^{(\kappa)}$ with increasing accuracy.
At each iteration, a refined outer approximation $\mathcal P^{(\kappa)}$, of $\mathcal U$ is generated, such that $\mathcal P^{(1)}\supset \mathcal P^{(2)} \supset \cdots \supset\mathcal U$. Let $\mathcal Z^{(\kappa)}$ denote the set containing all the vertices of the polyblock $\mathcal P^{(\kappa)}$.
Since the optimal $\mathbf w$ must be on the Pareto boundary of $\mathcal U$, we will try to find that point in a shrinking search region.
The vertex that achieves the maximum SNR in $\mathcal Z^{(\kappa)}$ is defined by $\tilde { \mathbf z}^{(\kappa)}$, i.e.,
$\tilde { \mathbf z}^{(\kappa)}=\arg \max _{\mathbf z\in \mathcal Z^{(\kappa)}} {\bf SNR}(\mathbf z)$, which is chosen for determining the next Pareto boundary point on $\mathcal U$. Define $\lambda \tilde {\mathbf z}^{(\kappa)}$ as the line that connects the points $\mathbf 0$ and $\tilde {\mathbf z}^{(\kappa)}\triangleq[\tilde z_1^{(\kappa)}, \cdots, \tilde z_{N_T}^{(\kappa)}]^T$. Then the next feasible point $\mathbf w^{(\kappa)}\triangleq[w_1^{(\kappa)}, \cdots, w_{N_T}^{(\kappa)}]^T$ is computed as the intersection point on the Pareto boundary of ${\mathcal U}$ with the line $\lambda \tilde {\mathbf z}^{(\kappa)}$. The following method is used to generate $N_T$ new vertices adjacent to $\tilde {\mathbf z}^{(\kappa)}$.
\ben\label{equ:39}
\mathbf z^{(\kappa),i}=\tilde {\mathbf z}^{(\kappa)}-(\tilde z_i^{(\kappa)}-w_i^{(\kappa)})\mathbf e_i, \,i=1, \cdots, N_T,
\enn
where $\mathbf z^{(\kappa),i}$ denotes the $i$th new vertex generated at the $\kappa$th iteration. Then the new vertex set can be expressed as
\ben\label{equ:40}
\mathcal Z^{(\kappa+1)}=\Big(\mathcal Z^{(\kappa)}\backslash \tilde {\mathbf z}^{(\kappa)}\Big)\cup \{\mathbf z^{(\kappa),1}, \cdots, \mathbf z^{(\kappa),N_T}\}.
\enn
Each vertex $\mathbf z\in \mathcal Z^{(\kappa+1)}$ defines a box, and thus the new polyblock $\mathcal P^{(\kappa+1)}$ is the union of all these boxes.  The upper and lower bound are refined as follows. The current upper bound is $f_{\max}^{(\kappa+1)}=\max_{\mathbf z\in \mathcal Z^{(\kappa+1)}}{\bf SNR}(\mathbf z)$ and the current lower bound is the maximum SNR among all the feasible points found so far: $f_{\min}^{(\kappa+1)}=\max_{\kappa}{\bf SNR}(\mathbf w^{(\kappa)})$.
The algorithm terminates when the gap between $f_{\min}^{(\kappa+1)}$ and $f_{\max}^{(\kappa+1)}$ meets some threshold. The optimal $\mathbf w$ is the feasible point $\mathbf w^{(\kappa)}$ that achieves $f_{\min}^{(\kappa+1)}$.

Now, the only remaining problem is how to determine the intersection point $\mathbf w^{(\kappa)}$, which will be addressed next.

\subsection{Finding Intersection Points by the Rate Profile Approach}
In this subsection, we show how to determine the intersection point $\mathbf w^{(\kappa)}$ on the Pareto boundary of $\mathcal U$ with the line $\lambda \tilde{\mathbf z}^{(\kappa)}$, to apply PA Algorithm. To proceed, we first introduce the following lemma, which is important for obtaining $\mathbf w^{(\kappa)}$.
\begin{lemma}\label{lemma:pareto}
For any $\mathbf w$ on the Pareto boundary of $\mathcal U$, the corresponding $\mathbf G$ satisfies $\text{tr}(\mathbf G)=P_s$.
\end{lemma}
\begin{proof}
Suppose that
\ben
\mathbf w^{\sharp}\triangleq [\text{tr}(\mathbf H_1^H\mathbf H_1\mathbf G^\sharp), \cdots, \text{tr}(\mathbf H_R^H\mathbf H_R\mathbf G^\sharp)]^T\nonumber
\enn
is on the Pareto boundary of $\mathcal U$. If $\text{tr}(\mathbf G^{\sharp })<P_s$, we can scale $\mathbf G^{\sharp}$ to $\mathbf G'$ such that $\mathbf G'=\beta\mathbf G^{\sharp}$ for some $\beta>1$, and $\text{tr}(\mathbf G^{\sharp })< \text{tr}(\mathbf G')\leq P_s$.  Then $\mathbf w'\triangleq [\text{tr}(\mathbf H_1^H\mathbf H_1\mathbf G'), \cdots, \text{tr}(\mathbf H_R^H\mathbf H_R\mathbf G')]^T > \mathbf w^{\sharp}$, which contradicts to the assumption that $\mathbf w^{\sharp}$ is on the Pareto boundary of $\mathcal U$. Therefore we have $\text{tr}(\mathbf G^{\sharp })=P_s$.
\end{proof}
\begin{figure}[!t]
\centering
\includegraphics[width=2.5in]{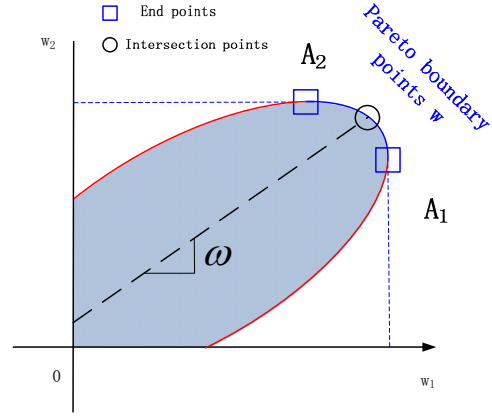}
\caption{An example of $\mathcal U$ when $R=2$.
    The Pareto boundary is only a part of the  boundary of $\mathcal U$. Two end points $A_i$ is determined by $\mathbf G_{\mathbf A_i}\triangleq\arg\max_{\text{tr}(\mathbf G)= P_s}\text{tr}(\mathbf H_i^H\mathbf H_i\mathbf G)=P_s{\bm\upsilon}(\mathbf H_i^H\mathbf H_i)[{\bm\upsilon}(\mathbf H_i^H\mathbf H_i)]^H$, for $i=1,2$. Then point $A_i=(P_s\|\mathbf H_1{\bm\upsilon}(\mathbf H_i^H\mathbf H_i)\|_2^2,P_s\|\mathbf H_2{\bm\upsilon}(\mathbf H_i^H\mathbf H_i)\|_2^2)$. The dashed line uniquely determines the ratio ${\bm \omega}$ between each element of the intersection point $\mathbf w$ and its $1-$norm $\|\mathbf w\|_1$}\label{pareto}
\end{figure}

Lemma \ref{lemma:pareto} states that any Pareto boundary point $\mathbf w\in \mathcal U$ must have its corresponding $\mathbf G$ satisfying $\text{tr}(\mathbf G)=P_s$.  As can be seen from Fig. \ref{pareto}, any point $\mathbf w\in  {\mathcal U}$ corresponds to a profile vector ${\bm \omega}\triangleq[\omega_1, \cdots, \omega_R]=\mathbf w/\|\mathbf w\|_1$, or equivalently, the slope of the line $\lambda \tilde{\mathbf z}^{(\kappa)}$. Consequently, the intersection point $\mathbf w^{(\kappa)}$ can be expressed as ${\bm \omega}Q^{\sharp}$, where $Q^{\sharp}$ is the optimal value of the following problem:
\begin{subequations}\label{equ:normprofile}
\ben
\max_{\mathbf G, Q}&& Q\\
\text{s.t.}&& \text{tr}(\mathbf H_i^H\mathbf H_i\mathbf G)= \omega_iQ, i=1, \cdots, R,\\
&&\text{tr}(\mathbf G)= P_s,\\
&&\mathbf G\succeq \mathbf 0.
\enn
\end{subequations}
The above approach to find $\mathbf w^{(\kappa)}$ is known as \emph{rate profile}\cite{LL_ploy}.

\eqref{equ:normprofile} is an SDP problem and can be efficiently solved using the MATLAB tool package such as CVX \cite{CVX}. Denote the optimal solution as $\mathbf G^{(\kappa)}$.  According to Proposition~\ref{pro:monotonic}, $\mathbf G^{(\kappa)}$ must be of rank one. Then the intersection point $\mathbf w^{(\kappa)}=[\text{tr}(\mathbf H_1^H\mathbf H_1\mathbf G^{(\kappa)}), \cdots, \text{tr}(\mathbf H_R^H\mathbf H_R\mathbf G^{(\kappa)})]^T$, and the corresponding $\mathbf g^{(\kappa)}$ is obtained by eigenvalue decomposition of $\mathbf G^{(\kappa)}$ as $\mathbf G^{(\kappa)}=\mathbf g^{(\kappa)}\mathbf g^{(\kappa)H}$.

\begin{remark}
It should be mentioned that we adopted a different approach compared to  \cite{LL_ploy} for determining the feasible point $\mathbf w^{(\kappa)}$. In \cite{LL_ploy}, the solution involves iterations between a bisection algorithm and an SDP problem. Our work, however,  presents direct approach for obtaining $\mathbf w^{(\kappa)}$, hence bypassing any bisection approach.
\end{remark}

\subsection{The Overall Algorithm for Determining Global Optimal $\mathbf g$}
The PA-based algorithm for solving \eqref{equ:boxnewu} is summarized as Algorithm \ref{alg:PA}.
\begin{table}[!h]
\centering \caption{Algorithm II: PA-based Algorithm for determining the optimal $\mathbf g$}
\begin{tabular}{|p{0.5cm}|p{7cm}|cl}\hline
1 &Set $\kappa=0$ and $\delta_2$ as the given threshhold. Initialize $\mathcal Z^{(\kappa)}=\{\mathbf b_0\}$, $\tilde {\mathbf z}^{(\kappa)}=\mathbf b_0$, $\mathcal V_{z}^{(\kappa)}=\{y|y={\bf SNR}(\mathbf z), \,\mathbf z\in \mathcal Z^{(\kappa)}\}$,
$f_{\min}^{(\kappa)}=0$, and $f_{\max}^{(\kappa)}=\max_{y\in\mathcal
V_z^{(\kappa)}}y$, where $\mathbf b_0\triangleq[P_s\lambda_{\max}(\mathbf H_1^H\mathbf H_1),\cdots, P_s\lambda_{\max}(\mathbf H_R^H\mathbf H_R)]^T$.    The initial $\mathbf g^{(\kappa)}$ is the nonrobust beamforming vector in \cite{AF-BF}.\\\hline
2 &Compute the intersection point $\mathbf w^{(\kappa)}$ on the Pareto boundary of $\mathcal U$ with the line $\lambda\tilde {\mathbf z}^{(\kappa)}$ and obtain the corresponding $\mathbf g^{(\kappa)}$.\\\hline
3&Compute $N_T$ new vertices that are adjacent to $\mathbf w^{(\kappa)}$ by \eqref{equ:39} and update $\mathcal Z^{(\kappa+1)}$ by \eqref{equ:40}. Let $\mathcal V^{(\kappa+1)}=\left\{\mathcal V^{(\kappa)}\backslash {\bf SNR}(\tilde {\mathbf z}^{(\kappa)})\right\}\cup\{{\bf SNR}(\mathbf z^{(\kappa),i})\},i=1, \cdots, N_T.$\\\hline
4&Update the lower bound and upper bound $f_{\min}^{(\kappa+1)}=\max_{\kappa}{\bf SNR}(\mathbf w^{(\kappa)})$, $f_{\max}^{(\kappa+1)}=\max_{y\in \mathcal V_z^{(\kappa+1)}}y$. Let $\kappa_0\triangleq \arg\max_{\kappa}{\bf SNR}(\mathbf w^{(\kappa)})$ and $\tilde {\mathbf z}^{(\kappa+1)}$ be the associate $\mathbf z\in \mathcal Z^{(\kappa+1)}$ that achieves $f_{\max}^{(\kappa+1)}$.\\\hline
5&If $ f_{\max}^{(\kappa+1)}-f_{\min}^{(\kappa+1)}
 \leq \delta_2$, go to Step $6$. Otherwise, let $\kappa=\kappa+1$, and go to Step $2$.\\\hline
6&Return $\mathbf g^{\sharp}=\mathbf g^{(\kappa_0)}$.\\\hline
\end{tabular}\label{alg:PA}
\end{table}

\subsection{Low-complexity Suboptimal Methods for Determining $\mathbf g$}\label{sec:robd}
 The optimal solution obtained from Algorithm \ref{alg:PA} is of  high complexity. In practice, it can be observed that computing the global optimal solution is practically feasible for a small number of relays.  Thus we treat Algorithm \ref{alg:PA} mainly as
a benchmark for performance evaluation. For practical implementation, in this subsection,
we propose two low-complexity suboptimal methods for determining the source BF vector $\mathbf g$, which provides a tradeoff between the computational complexity and the system performance.
\subsubsection{Robust gradient method}
The first method applies the gradient method in \cite[Table I]{AF-BF} with grad$_{\bar{\mathbf g}}$ determined by the following gradient estimate
\ben
\text{grad}_{\bar {\mathbf g}}&=&\frac{1}{2\delta}\big[\big(\text{SNR}(\bar{\mathbf g}+\delta \mathbf e_1)-\text{SNR}(\bar{\mathbf g}-\delta \mathbf e_1)\big), \cdots,\nonumber\\ &&\big(\text{SNR}(\bar{\mathbf g}+\delta \mathbf e_{2N_T})-\text{SNR}(\bar{\mathbf g}-\delta \mathbf e_{2N_T})\big)\big]^T,\label{equ:gradient}
\enn
where $\bar{\mathbf g}\triangleq[\text{Re}\{\mathbf g\}^T, \text{Im}\{\mathbf g\}^T]^T$, $\delta$ is a small positive constant and the SNR in \eqref{equ:gradient} is expressed as a function of $\bar{\mathbf g}$. Note that our gradient estimate in \eqref{equ:gradient} is different from \cite{AF-BF}, where they compute it in an analytical form for each $\bar{\mathbf g}$. By comparison, for evaluating the gradient estimate in \eqref{equ:gradient}, it has to apply Algorithm~\ref{alg:dinkelbach} for all $4N_T$ vectors $\bar{\mathbf g}_k\pm \delta \mathbf e_i$, $1\leq i\leq 2N_T$. As will be seen in section~\ref{sec:simu}, this method preserves the optimality to some extent.
\subsubsection{Simplified robust method}
In this method, We choose $\mathbf g$ as the nonrobust solution in \cite{AF-BF}, and the power allocation factor $\mathbf c$ as the solution of \eqref{equ:robbw} for given $\mathbf g$. Since this method utilizes Algorithm~\ref{alg:dinkelbach} only once, it has much lower complexity than the \emph{Robust gradient method}. However, as verified in section~\ref{sec:simu}, it shows a near-optimal performance.
\section{Implementation Issues and Complexity}\label{subsec:design}
In this section, we discuss implementation issues and computational complexity for the proposed algorithms. For computing the source BF vector $\mathbf g$, the source needs the CSI $\mathbf H_i$ of the first hops, and the available channel magnitudes $\|\tilde {\mathbf f}_i\|_2$ of the second hops, which can be fed back by each relay. After computing $\mathbf g$ and the real-valued optimal power allocation factor $\mathbf c^\sharp$ at the source, they will be broadcasted to each relay node.
For determining the relay BF matrices, each relay node only requires the local CSIs and the $\mathbf g$ and $c_i^\sharp$ from the source.

In Algorithm \ref{alg:dinkelbach}, one only needs to determine a vector with $R$ real variables $c_i$ rather than $R$ matrices $\mathbf B_i\in\mathbb C^{M_i\times M_i}$ in the conventional method \cite{nonrob}. According to \cite{socp_complexity}, the design complexity of solving the SOCP problem \eqref{equ:f4} can be approximated as $\mathcal O((2^R)^{\frac{3}{2}}R^3\log(1/\theta))$, given a solution accuracy $\theta>0$. Hence the complexity of Algorithm \ref{alg:dinkelbach} is $\mathcal O((2^R)^{\frac{3}{2}}R^3\log(1/\theta))$ that times the number of iterations. By contrast, using the method in \cite{nonrob}, the complexity of the SDP solver is $\mathcal O((N^2(N^2+1)/2)^3\log(1/\theta))$  with $N=\sum_{i=1}^RM_i^2$ that times the number of iterations, which is fairly high. As can be seen from above, our complexity of the SOCP problem in each iteration is much lower than that in \cite{nonrob}. In section \ref{sec:simu}, we will further show that the iteration number by the Dinkelbach-based algorithm is less than the bisection-based algorithm in \cite{nonrob}.

The major computing step of Algorithm \ref{alg:PA} in iteration $\kappa$ is solving problem \eqref{equ:normprofile} for determining $\mathbf w^{(\kappa)}$ and computing $N_T+1$ worst case ${\bf SNR}(\mathbf z)$,  including the intersection point $\mathbf w^{(\kappa)}$, and $N_T$ new vertices. According to \cite{Luo-SDP}, the complexity of solving SDP problem \eqref{equ:normprofile} can be approximated as $\mathcal O(\max (N_T, R+1)^4\sqrt{N_T}\log(1/\theta))$. Notice that in the perfect CSI case, ${\bf SNR}(\mathbf z)$ is directly obtained by Lemma~\ref{theorem:perfectCSI} and Corollary~\ref{cor:Jing}; in the robust case, ${\bf SNR}(\mathbf z)$ is obtained by Algorithm \ref{alg:dinkelbach} in section \ref{subsec:opRBFb}. Section \ref{sec:simu} shows the average iteration time of Algorithm \ref{alg:dinkelbach} and Algorithm \ref{alg:PA}.

\section{Simulations and Discussion}\label{sec:simu}

In this section, we provide numerical results to validate the proposed algorithms in this paper, using the numerical convex optimization solver CVX\cite{CVX}. First, the convergence of Algorithm~\ref{alg:dinkelbach} and Algorithm~\ref{alg:PA} is illustrated, comparing with the bisection approach and the BRB algorithm, respectively. Then, the performance evaluation of our robust design is addressed.

The channel fading is modeled as Rayleigh fading, and each channel entry satisfies the complex normal distribution $\mathcal C\mathcal N(0,1)$. The noise at each node is assumed to be zero-mean unit variance complex Gaussian random variables. We set the power consumed at the source as $10$dB. In our simulations, we set $\varepsilon_i$ as $\varepsilon_i^2=\rho\|\tilde {\mathbf f}_i\|_2^2$  with $\rho\in [0, 1)$. The larger $\rho$ is, the poorer CSI quality will be.   We also  set the convergence thresholds of Algorithm~\ref{alg:dinkelbach}, Algorithm~\ref{alg:PA} respectively as $\delta_1=0.01$, and $\delta_2=0.1$. All results are averaged over $100$ channel realizations.

 The following benchmarks are compared through simulations in this section. 
a) \emph{Perfect optimal method}: this is obtained by our proposed method in section~\ref{sec:opSBF} under perfect CSI assumption.
b) \emph{Perfect gradient method}: this is obtained by the gradient method in \cite{AF-BF} under perfect CSI assumption.
c) \emph{Robust optimal method}: The robust optimal design method proposed in Algorithm~\ref{alg:PA}.
d) \emph{Robust gradient method}.
e) \emph{Simplified robust method}.
f) \emph{Non-robust method}: this was proposed in \cite{AF-BF} using imperfect CSI.

\subsection{Convergence Evaluation}
\begin{figure}[h]
    \centering
    \includegraphics[width=3.5in]{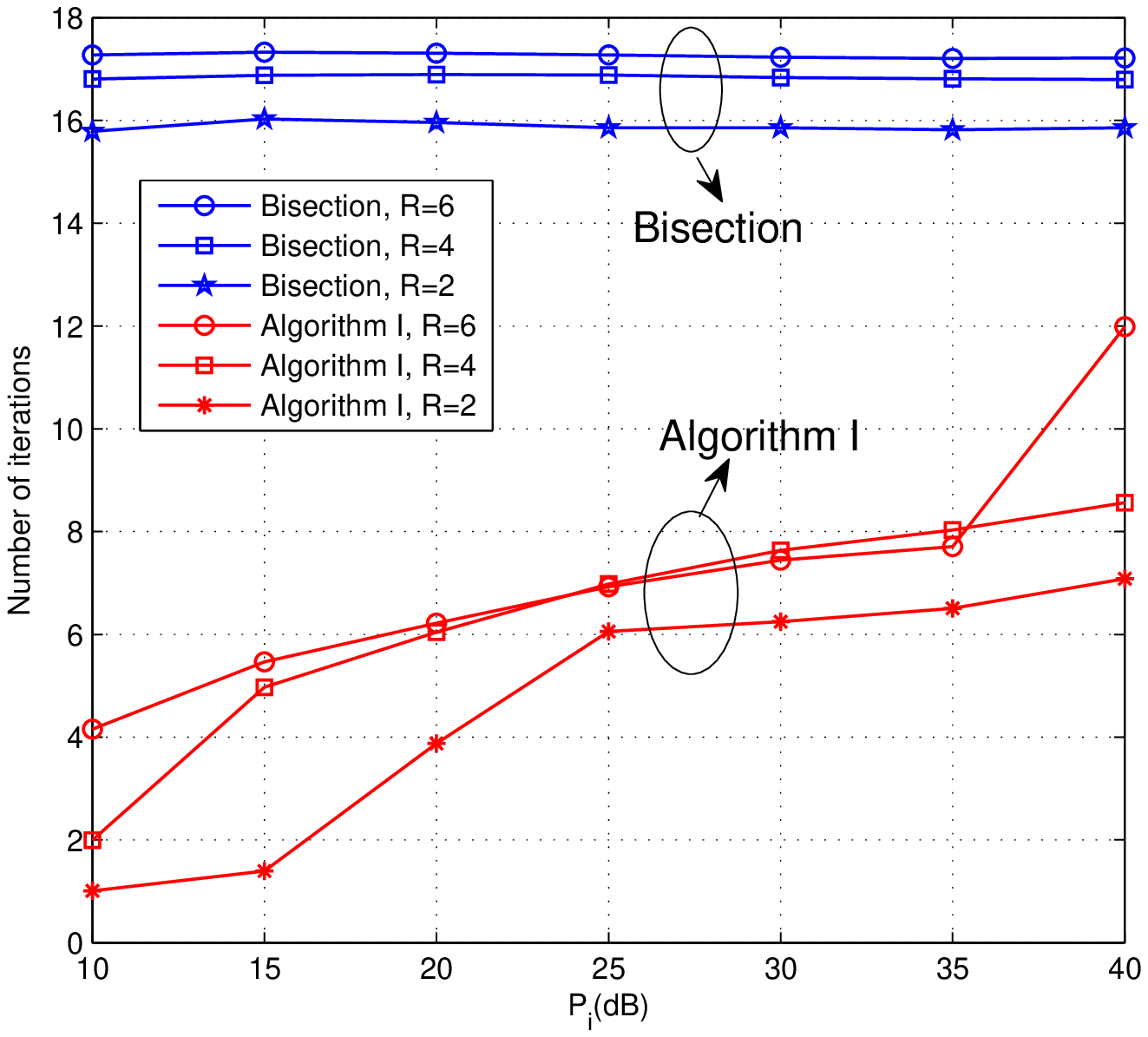}
    \caption{Average iteration time comparison for Algorithm \ref{alg:dinkelbach} and the bisection approach.}\label{fig:dinkelbach}
    \includegraphics[width=3.5in]{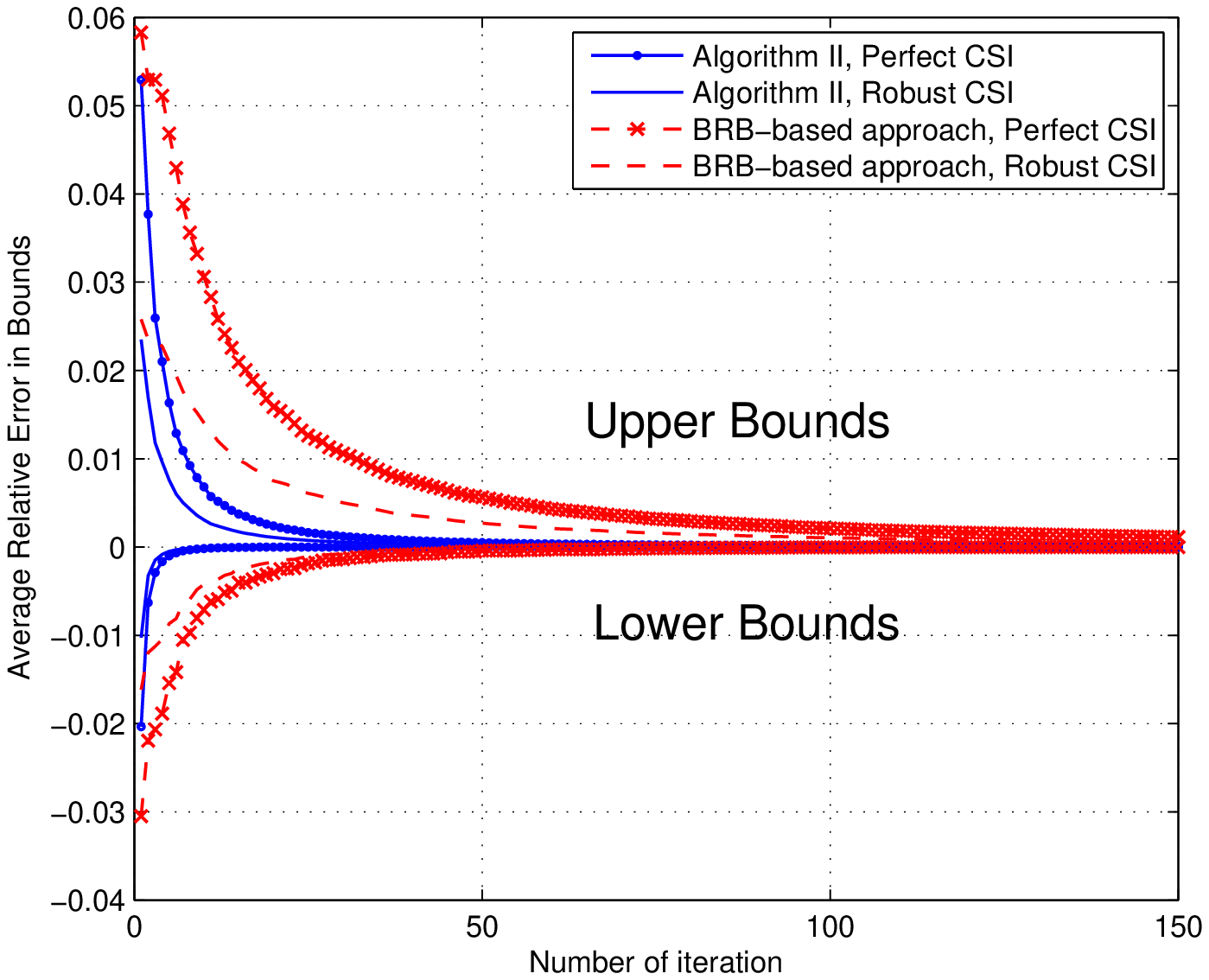}
    \caption{Relative error of lower and upper bounds on the SNR value versus the number of iteration.}\label{fig:PAnBRB}
\end{figure}
\begin{figure}[!thbp]
    \centering
    \includegraphics[width=3.5in]{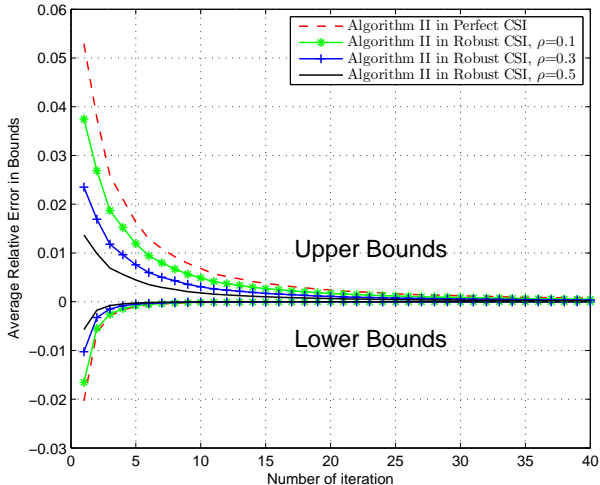}
    \caption{Relative error of lower and upper bounds on the SNR value versus the number of iterations for different $\rho$.}\label{fig:PAversusrho}
\end{figure}
Firstly, we study the convergence performance of Algorithm \ref{alg:dinkelbach}.  Fig. \ref{fig:dinkelbach} shows the average iteration time of Algorithm~\ref{alg:dinkelbach} and the Bisection approach to achieve the predefined accuracy $\delta_1$ for $R\in \{2,4,6\}$. The initial upper bound $\gamma_u^{(0)}$ and lower bound $\gamma_l^{(0)}$ of the Bisection approach are specified as the worst case received SNR of the \emph{Perfect optimal method} and that of the \emph{Non-robust method}, respectively. It can be observed that Algorithm~\ref{alg:dinkelbach} takes less than half iteration numbers of the Bisection approach for most of the SNR regime. Thus, Algorithm~\ref{alg:dinkelbach} is more efficient.

Then, we evaluate the convergence behavior of Algorithm~\ref{alg:PA} and the BRB-based algorithm in~\cite{EB_BRB}.  We set $(N_T, M_1, M_2)=(2,2,2)$, and fix the relay power as $30$dB. Fig.~\ref{fig:PAnBRB} shows the average iteration numbers to achieve the certain accuracies both in the perfect case and in the robust case, where we set $\rho=0.3$. The accuracy of the lower and upper bound are defined as $(f_{\min}-f_{\text{opt}})/f_{\text{opt}}$ and $(f_{\max}-f_{\text{opt}})/f_{\text{opt}}$, respectively, where $f_{\text{opt}}$ is the optimal value of the worst case SNR.  It can be seen that both algorithms quickly achieve the optimal solutions, but more iterations of the BRB algorithm is needed to achieve a certain accuracy. Thus we claim that in our problem, Algorithm~\ref{alg:PA} is more efficient than the BRB-based algorithm. Notice that the convergence performance of the  BRB and PA algorithm is also illustrated in \cite{EB_BRB} \cite{EB_BRB2}, showing that different algorithms are superior in different scenarios.

Another observation from Fig.~\ref{fig:PAnBRB} is that in the robust case both Algorithm~\ref{alg:PA} and BRB algorithm converges more quickly than that in the perfect case. This phenomena is further illustrated in Fig.~\ref{fig:PAversusrho}, which compares the lower and upper bound of the proposed PA algorithm under different $\rho$ assumption. It can be seen that the larger $\rho$ leads to a smaller gap between the upper bound and the lower bound in each iteration. This can be explained as the maximum value over the vertices of the polyblock $\mathcal P^{(\kappa)}$ is lower for larger $\rho$.
%


\subsection{Performance Comparison with the Existing Schemes}\label{subsec:providing}
\begin{figure}[h]
    \centering
    \includegraphics[width=3.5in]{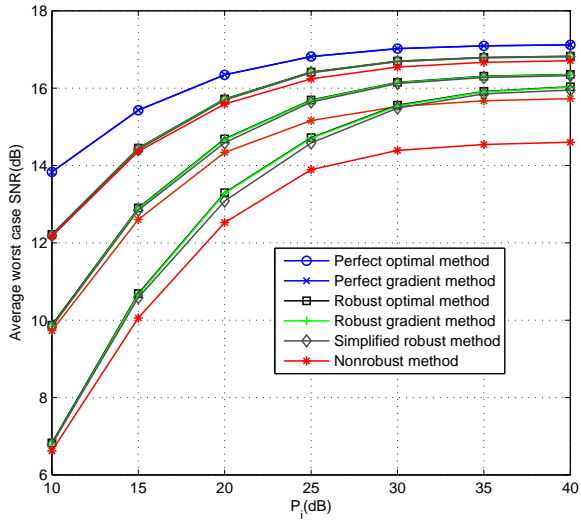}
    \caption{Average worst case SNR  versus different relay power in different error bound case}\label{fig:relay_two}
\end{figure}

We now compare our robust BF design with some existing schemes. The parameters are set as $(N_T, M_1, M_2)=(2,2,2)$. Fig. \ref{fig:relay_two} shows the average worst-case received SNR versus individual relay powers.   Simulations reveal that the \emph{Nonrobust method} will cause increasing performance loss with the increment of channel uncertainty,  compared to the perfect CSI case. Even when the relay power is very large, this loss cannot be compensated. It can be observed that when the relay power is $40$dB, and the channel uncertainty ratio $\rho=0.5$, this performance degradation is about $2.5$dB. On the other hand, the robust design can improve the performance for any channel uncertainty ratio. Although gradient method only achieves local optimality in theory, it behaves well in our simulations and has a close-to-optimal performance in both the perfect case and the robust case. It can also be seen from Fig.~\ref{fig:relay_two} that, as a simple yet efficient method, the \emph{Simplified robust method} has a near optimal performance, which greatly facilitates the practical application of the robust design.

\subsection{Performance Evaluation with Different Network Configurations}
\begin{figure}[!htp]
    \centering
    \includegraphics[width=3.5in]{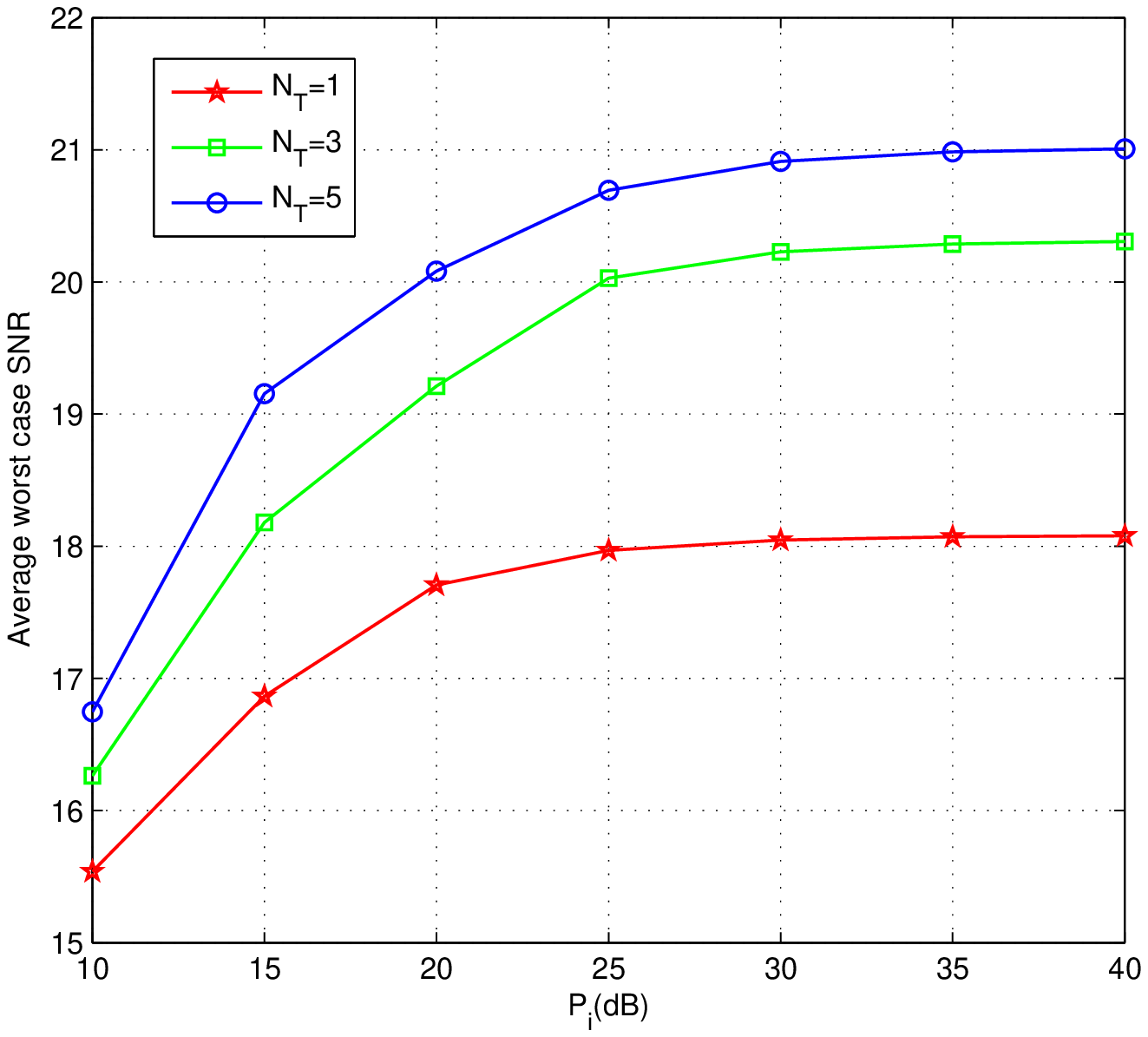}
    \caption{Average worst case SNR vs. different relay power with different numbers of transmit antennas.}\label{fig:source_num}
    \includegraphics[width=3.5in]{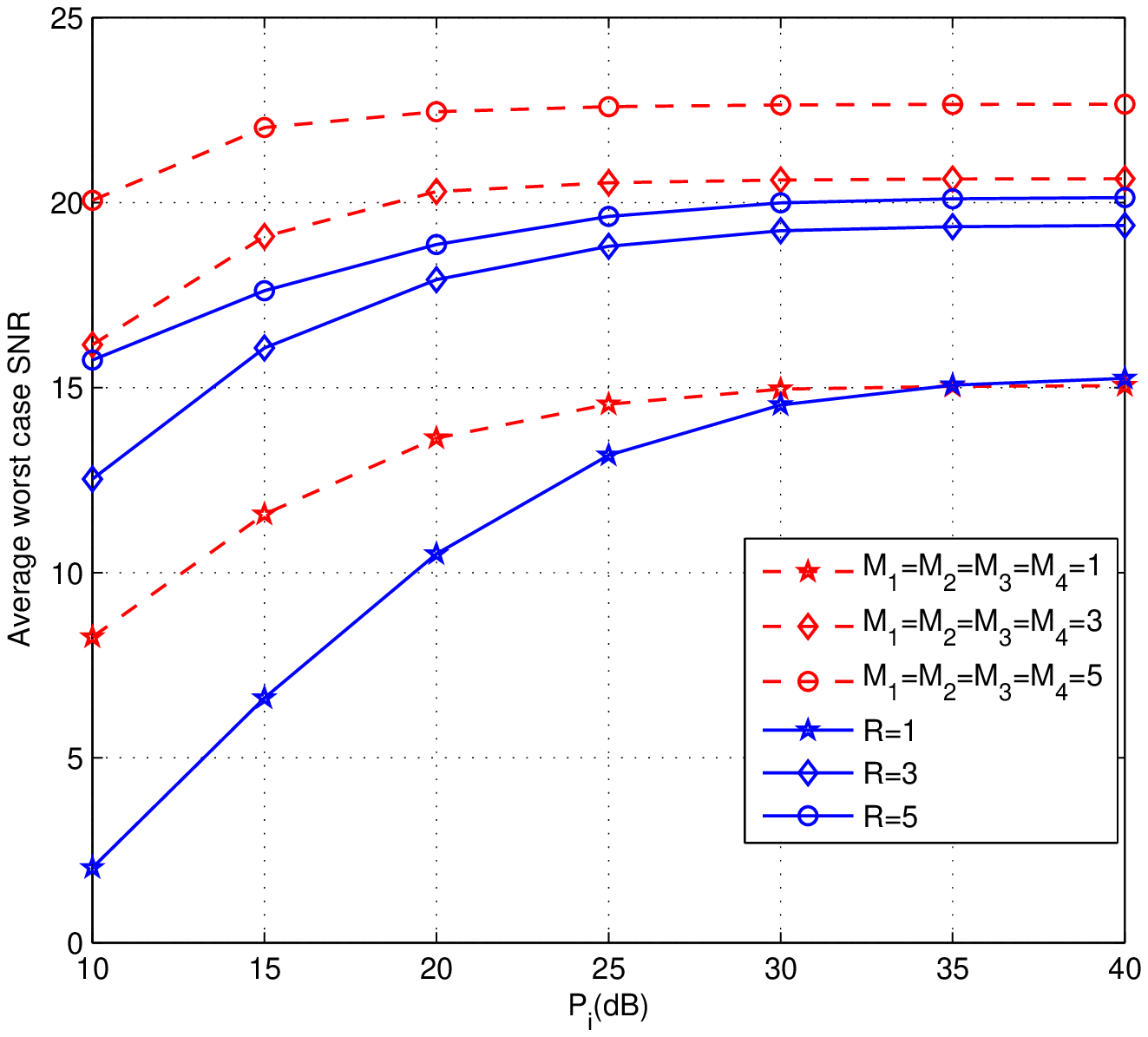}
    \caption{Average worst case SNR vs. different relay power with different relay number and different relay antennas.}\label{fig:relay_num}
\end{figure}
We investigate the impact of different network configurations. We set $\rho=0.3, R=3$ and $M_i=3$ for $i=1,2,3$. One can see that increasing the source antennas $N_T$ from $1$ to $3$ brings the most benefit to the system, and the SNR improves $2.3$ dB when the relay power is $40$ dB. However, the improvement is not so apparent if the source antenna number further increases to $5$, where the SNR only improves $0.7$ dB. The results in Fig. \ref{fig:source_num} indicate that small increment of the source antenna number can greatly improve the system performance.

Fig. \ref{fig:relay_num} compares the average worst case SNR versus the individual relay powers by the proposed robust design in \cite{ganzhengtsp}\cite{HS_Worst}. Here we assume $N_T=1$, since the method in \cite{ganzhengtsp}\cite{HS_Worst} cannot be applied to the general case $N_T\geq 2$. The channel uncertainty parameter $\rho$ is set to be $0.3$. When we consider the network with $R=4$ relays, we consider the cases $M_i=1,3,5$ respectively for $i=1,2,3,4$. The $M_i=1$ case corresponds to the method proposed in \cite{ganzhengtsp}.  From Fig. \ref{fig:relay_num}, when the antenna number at each relay increases from $1$ to $3$, the average worst case SNR increases $5.5$ dB. By contrast, when this number further increases to $5$, the average worst case SNR increases $2.0$ dB. We next investigate the impact of relay number in the system. Here, we assume that each relay is equipped with $2$ antennas. The $R=1$ curve corresponds to the method used in \cite{HS_Worst}. Similar result can be observed when we vary the relay number from $1$ to $3$, where the SNR increases about $4.1$ dB. If we further increase the relay number to $5$, the SNR increases about $0.7$ dB. Fig. \ref{fig:relay_num} shows that increasing the relay number and relay antenna number are both beneficial. Moreover, one can greatly improve the system performance by slightly increasing the relay number or relay antenna number, which validated the importance of our work.

\section{Conclusions}

In this paper, we consider a multi-antenna multi-relay channel with one source and one destination. Assuming that the relay only amplifies and forwards its received signals, we present a global optimal BF design in the robust case.  To maximize the worst case received SNR, we aim to  jointly design the BF matrices at the source and the relays under individual power constraints at the source and the relays. We give a semi-closed form of the relay BF matrices up to a power scalar factor. The optimal and suboptimal algorithms for solving the source BF vector are also proposed. Numerical results verify the advantage of the proposed algorithm over the existing methods.

\appendices
\section{Proof of Lemma \ref{lemma:robb1}}\label{app:lemmarobb1}
Suppose that the SVD of $\mathbf u_i$ is
\ben
\mathbf u_i=\mathbf U_i\begin{bmatrix}\|\mathbf u_i\|_2\\ \mathbf 0_{M_i-1}\end{bmatrix}\triangleq\mathbf U_i{\bm \Sigma_i},\label{equ:ui}
\enn
where  the unitary matrix $\mathbf U_i\in \mathbb C^{M_i\times M_i}$. Then we can express  the relay BF matrices as
\ben
\mathbf B_i=\mathbf Y_i\mathbf U_{i}^H\label{equ:bi},
\enn
where $\mathbf Y_i\in \mathbb C^{M_i\times M_i}$ is a matrix to be determined. Upon substituting \eqref{equ:bi} and \eqref{equ:ui} into \eqref{equ:robb1}, the max-min SNR problem subject to the individual power constraints is given by
\begin{subequations}\label{equ:sub_bi}
\ben
\max _{\mathbf Y_i}  \min_{\triangle \mathbf f\in \mathcal A}&&\frac{|\sum_{i=1}^R\mathbf f_i^T\mathbf Y_i\mathbf {\Sigma}_{i}|^2}{\sigma_D^2+\sigma_R^2\sum_{i=1}^R\mathbf f_i^T\mathbf Y_i\mathbf Y_i^H\mathbf f_i^{\ast}},\\
\text{s.t}.&& \text{tr}(\mathbf Y_i(\mathbf {\Sigma}_{i}\mathbf {\Sigma}_{i}^H+\sigma_R^2)\mathbf Y_i^H)\leq P_i.
\enn
\end{subequations}
We can further partition $\mathbf Y_i$ as
\ben
\mathbf Y_i=\begin{bmatrix}\mathbf b_i &\mathbf Z_{yi}\end{bmatrix},\nonumber
\enn
where $\mathbf b_i\in \mathbb C^{M_i\times 1}$  and $\mathbf Z_{yi}\in \mathbb C^{M_i\times (M_i-1)}$.
Then we have
\ben
\mathbf Y_i\mathbf {\Sigma}_{i}&=&\begin{bmatrix}\mathbf b_i &\mathbf Z_{yi}\end{bmatrix}\begin{bmatrix}\|\mathbf u_i\|_2\\ \mathbf 0\end{bmatrix}=\|\mathbf u_i\|_2\mathbf b_i.\label{equ:22}
\enn
Upon substituting \eqref{equ:22} into \eqref{equ:sub_bi}, we have the received SNR at the destination as
\ben
\text{SNR}&=&\frac{|\sum_{i=1}^R\mathbf f_i^T\mathbf b_i\|\mathbf u_i\|_2|^2}{\sigma_D^2+\sigma_R^2\sum_{i=1}^R\|\mathbf f_i^T\mathbf Y_i\|_2^2},\nonumber\\
&=&\frac{|\sum_{i=1}^R\mathbf f_i^T\mathbf b_i\|\mathbf u_i\|_2|^2}{\sigma_D^2+\sigma_R^2\sum_{i=1}^R(\|\mathbf f_i^T\mathbf b_i\|_2^2+\|\mathbf f_i^T\mathbf Z_{yi}\|_2^2)}\label{equ:indp},
\enn
and the individual relay power becomes
\ben
&&\text{tr}(\mathbf Y_i[\mathbf {\Sigma}_{i}\mathbf {\Sigma}_{i}^H+\sigma_R^2]\mathbf Y_i^H)\nonumber\\
&=& \text{tr}(\mathbf b_i(\|\mathbf u_i\|_2^2+\sigma_R^2)\mathbf b_i^H)
+\sigma_R^2\text{tr}(\mathbf Z_{yi}\mathbf Z_{yi}^H)\nonumber\\
&=& (\|\mathbf u_i\|_2^2+\sigma_R^2)\|\mathbf b_i\|_2^2
+\sigma_R^2\text{tr}(\mathbf Z_{yi}\mathbf Z_{yi}^H).\nonumber
\enn
From \eqref{equ:indp}, to achieve maximum SNR with respect to $\mathbf Y_i$, we must minimize the denominator of SNR by forcing $\mathbf Z_{yi}=\mathbf 0$.
Then we can express $\mathbf B_i$ as
\ben\label{equ:exb}
\mathbf B_i=\mathbf b_i(\mathbf U_i)_1^H=\mathbf b_i\mathbf {\hat u}_{i}^H,
\enn
where $(\mathbf U_i)_1$ denotes the first column of $\mathbf U_i$. Substituting \eqref{equ:exb} into \eqref{equ:robb1}, we get \eqref{equ:robb4}.

\section{Proof of Lemma \ref{lemma:robb2}}\label{app:lemmarobb2}
When $\mathbf b_i=c_i\hat{\tilde {\mathbf f}}_{i}^{\ast}$, we have $\mathbf B_i=c_i\hat{\tilde {\mathbf f}}_{i}^{\ast}\mathbf {\hat u}_i^H$ by Lemma~\ref{lemma:robb1}. Then the objective function of \eqref{equ:robb1} becomes
\ben
&&\frac{\big|\sum_{i=1}^R(\tilde {\mathbf f}_i+\eta_i\tilde {\mathbf f}_{i}^{\|}+\tau_i\tilde {\mathbf f}_i^{\bot})^Tc_i\hat{\tilde {\mathbf f}}_{i}^{\ast}\|\mathbf u_i\|_2\big|^2}{\sigma_R^2\sum_{i=1}^R\|(\tilde  {\mathbf f}_{i}+\eta_i\tilde  {\mathbf f}_{i}^{\|}+\tau_i\tilde  {\mathbf f}_i^{\bot})^Tc_i\hat{\tilde {\mathbf f}}_{i}^{\ast}\|_2^2+\sigma_D^2}\nonumber\\
&=&\frac{\big|\sum_{i=1}^R(\|\tilde {\mathbf  f}_i\|_2+\eta_i)c_i\|\mathbf u_i\|_2\big|^2}{\sigma_R^2\sum_{i=1}^R|\|\tilde {\mathbf  f}_i\|_2+\eta_i|^2 |c_i|^2+\sigma_D^2}.\label{equ:inner}
\enn
where we have decomposed $\triangle {\mathbf f}_i=\eta_i\tilde  {\mathbf  f}_i^{\|}+\tau_i\tilde  {\mathbf f}_i^{\bot}$, with $|\eta_i|^2+| \tau_i|^2\leq \varepsilon_i^2$ and $\eta_i, \tau_i\in \mathbb C$. \eqref{equ:inner} implies that when $\mathbf b_i=c_i\hat{\tilde {\mathbf f}}_{i}^{\ast}$, only $\eta_i$ affects the minimum value of  \eqref{equ:inner}. Then we can focus on $\Delta\mathbf f_i=\eta_i\tilde {\mathbf f}_i^{\|}$, or $\mathbf f_i=(\|\tilde {\mathbf f}_i\|_2+\eta_i)\hat{\tilde {\mathbf f}}_{i}$, with $|\eta|_i\leq \varepsilon_i$.
Thus \eqref{equ:robb1} is equivalent to
\begin{subequations}\label{equ:q2n}
\ben
\max_{\mathbf c} \min _{|\eta_i|\leq \varepsilon_i}&&\frac{\big|\sum_{i=1}^R(\|\tilde {\mathbf f}_i\|_2+\eta_i)c_i\|\mathbf u_i\|_2\big|^2}{\sigma_R^2\sum_{i=1}^R|\|\tilde {\mathbf  f}_i\|_2+\eta_i|^2 |c_i|^2+\sigma_D^2},\\
\text{s.t.}&& |c_i|\leq \sqrt{\frac{P_i}{\sigma_R^2+\|\mathbf u_i\|_2^2}}.
\enn
\end{subequations}

It is worth noting that in \eqref{equ:q2n}, $\eta_i$ is a complex value. We will show in the following that \eqref{equ:q2n} can be transformed into a problem with real valued variable $\eta_i$, which is further limited to $\pm\varepsilon_i$. Our work comes from the idea of \emph{real-valued implemention} that has recently been proposed in \cite{LZ_Realvalue}.
Define
\ben
f_{\eta i}&\triangleq& \|\tilde {\mathbf f}_i\|_2+\eta_i,\nonumber\\
\mathbf R_s&\triangleq&(\mathbf u\odot \mathbf f_{\eta})^\ast(\mathbf u\odot \mathbf f_{\eta})^T,\nonumber\\
\mathbf R_n&\triangleq&\sigma_R^2\text{diag}[|\mathbf f_{\eta}|^2],\nonumber
\enn
where $\mathbf u\triangleq [\|\mathbf u_1\|_2,\cdots,\|\mathbf u_R\|_2]^T$ and the operator $\odot$ denotes the point-wise multiplication of two vectors.
Then we can write the objective of \eqref{equ:q2n} as
\ben\label{equ:snrrd}
\text{SNR}=\frac{\mathbf c^H\mathbf R_s\mathbf c}{\mathbf c^H\mathbf R_n\mathbf c+\sigma_D^2}.
\enn

Note that $\mathbf R_n$ is a real-valued diagonal matrix, while $\mathbf R_s$ is in general complex-valued. The \emph{real-valued implementation} idea~\cite{LZ_Realvalue} aims to transform $\mathbf R_s$ into a real-valued matrix.
First we can write $\mathbf u\odot \mathbf f_{\eta}=\mathbf u\odot |\mathbf f_{\eta}| \odot{\bm  \varphi}$, where ${\bm  \varphi}\triangleq [e^{j\varphi_{1}}, \cdots, e^{j\varphi_{R}}]^T$, $e^{j\varphi_{i}}$ denotes the phase of $f_{\eta i}$ and $j=\sqrt{-1}$. Then for any complex vector $\mathbf c$, one can always decompose it into the form $\mathbf c=\tilde {\mathbf c}\odot \tilde {\bm  \varphi}$, where $\tilde {\bm  \varphi}\triangleq [e^{-j\varphi_{1}}, \cdots, e^{-j\varphi_{R}}]^T$, and  $\tilde {\mathbf c}$ is determined by element-wise division between $\mathbf c$ and $\tilde {\bm  \varphi}$. By referring to \eqref{equ:snrrd}, the objective of \eqref{equ:q2n} is given by
\ben\label{equ:realvalue}
\text{SNR}&=&\frac{(\tilde {\mathbf c}\odot \tilde {\bm  \varphi})^H(\mathbf u\odot |\mathbf f_{\eta}| \odot{\bm  \varphi})^\ast(\mathbf u\odot |\mathbf f_{\eta}| \odot{\bm  \varphi})^T(\tilde {\mathbf c}\odot \tilde {\bm  \varphi})}{(\tilde {\mathbf c}\odot {\bm  \varphi})^H\mathbf R_n(\tilde {\mathbf c}\odot {\bm  \varphi})+\sigma_D^2}\nonumber\\
&=& \frac{\tilde {\mathbf c}^H(\mathbf u\odot |\mathbf f_{\eta}| )^\ast(\mathbf u\odot |\mathbf f_{\eta}| )^T\tilde {\mathbf c}}{\tilde {\mathbf c}^H\mathbf R_n\tilde {\mathbf c}+\sigma_D^2}\nonumber\\
&=&\frac{\tilde {\mathbf c}^H\bar{\mathbf R}_s\tilde {\mathbf c}}{\tilde {\mathbf c}^H\mathbf R_n\tilde {\mathbf c}+\sigma_D^2}\nonumber,
\enn
where $\bar{\mathbf R}_s\triangleq(\mathbf u\odot |\mathbf f_{\eta}| )^\ast(\mathbf u\odot |\mathbf f_{\eta}| )^T$ is a real-valued matrix. Notice for any real-valued $\bar{\mathbf R}_s$, $\mathbf R_n$, by maximizing the received SNR, the corresponding $\tilde {\mathbf c}$ must be real-valued \cite{ganzhengtsp} \cite{LZ_Realvalue}. Now \eqref{equ:q2n} can be rewritten as
\begin{subequations}\label{equ:phi}
\ben
\max_{\tilde {\mathbf c}\in \mathbb R^R} \min _{|\eta_i|\leq \varepsilon_i}&&\frac{\tilde {\mathbf c}^H\bar{\mathbf R}_s\tilde {\mathbf c}}{\tilde {\mathbf c}^H\mathbf R_n\tilde {\mathbf c}+\sigma_D^2}\\
\text{s.t.}&& |\tilde c_i|\leq \sqrt{\frac{P_i}{\sigma_R^2+\|\mathbf u_i\|_2^2}}.
\enn
\end{subequations}
\begin{figure}[!t]
    \centering
    \includegraphics[width=3.5in]{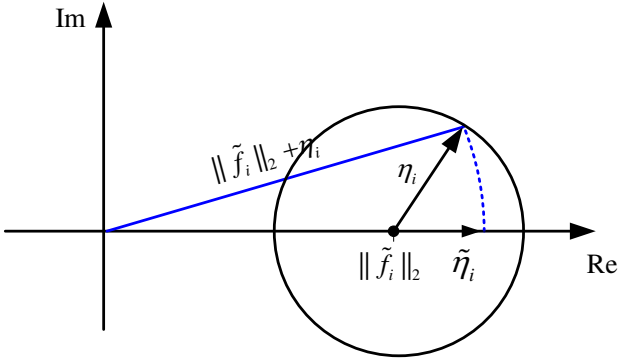}
    \caption{Rotating $\|\tilde{ \mathbf f}_i\|_2+\eta_i$ to the real axis, one can always find a real valued $\tilde {\eta}_i$ such that $|\|\tilde{\mathbf  f}_i\|_2+\eta_i|=|\|\tilde {\mathbf f}_i\|_2+\tilde {\eta}_i|$.}\label{fig:rot}
\end{figure}
In  \cite{LZ_Realvalue}, since ${\bm  \varphi}$ is fixed in their perfect CSI assumption,  they find the optimal solution $\tilde {\mathbf c}$ in \eqref{equ:phi} and obtain $\mathbf c$ by $\mathbf c=\tilde {\mathbf c}\odot \tilde{\bm  \varphi}$. In our case, things are a bit different: the value of $\tilde{\bm \varphi}$ is not important in this problem.
From \eqref{equ:phi}, we can see that it is $|\mathbf f_{\eta}|$ rather than $\tilde{\bm \varphi}$ that affects  the value of worst case SNR. For any given $|\eta_i|\leq \varepsilon_i, \eta_i\in \mathbb C$, we can find a real-valued $|\tilde {\eta}_i|\leq \varepsilon_i, \tilde {\eta}_i\in \mathbb R$ such that $|\|\tilde{\mathbf f}_i\|_2+\eta_i|=|\|\tilde {\mathbf f}_i\|_2+\tilde {\eta}_i|$ as shown in Fig. \ref{fig:rot}. Therefore, considering real-valued $\eta_i$ will not lose the optimality of \eqref{equ:phi}, or equivalently \eqref{equ:q2n}. By slight abuse of notation $\mathbf c$ instead of $\tilde {\mathbf c}$,  we can transform \eqref{equ:q2n} into
\begin{subequations}\label{equ:phi2}
\ben
\max_{\mathbf c} \min _{-\varepsilon_i\leq\eta_i\leq \varepsilon_i}&&\frac{\left(\sum_{i=1}^Rf_{\eta i}c_i\|\mathbf u_i\|_2\right)^2}{\sigma_R^2\sum_{i=1}^Rf_{\eta i}^2 c_i^2+\sigma_D^2},\\
\text{s.t.}&& c_i\leq \sqrt{\frac{P_i}{\sigma_R^2+\|\mathbf u_i\|_2^2}},
\enn
\end{subequations}

Introducing a slack variable $\gamma$, problem \eqref{equ:phi2} is transformed into
\begin{subequations}\label{equ:q4nwen}
\ben
\max_{\mathbf c, \gamma}&& \gamma\\
\text{s.t.}&&\min_{-\varepsilon_i\leq\eta_i\leq \varepsilon_i}\frac{\left(\sum_{i=1}^Rf_{\eta i}c_i\|\mathbf u_i\|_2\right)^2}{\sigma_R^2\sum_{i=1}^Rf_{\eta i}^2c_i^2+\sigma_D^2}\geq \gamma,\label{equ:boxnbwen}\\
&& c_i\leq \sqrt{\frac{P_i}{\sigma_R^2+\|\mathbf u_i\|_2^2}},
\enn
\end{subequations}
which is equivalent to the following problem.
\begin{subequations}\label{equ:q4n}
\ben
\max_{\mathbf c, \gamma}&& \gamma\\
\text{s.t.}&&\frac{\left(\sum_{i=1}^Rf_{\eta i}c_i\|\mathbf u_i\|_2\right)^2}{\sigma_R^2\sum_{i=1}^Rf_{\eta i}^2c_i^2+\sigma_D^2}\geq \gamma, \,\, -\varepsilon_i\leq\eta_i\leq \varepsilon_i,\label{equ:boxnb}\\
&& c_i\leq \sqrt{\frac{P_i}{\sigma_R^2+\|\mathbf u_i\|_2^2}}.
\enn
\end{subequations}
Let
\ben
f(\mathbf f_{\eta})\triangleq -\sum_{i=1}^Rf_{\eta i}c_i\|\mathbf u_i\|_2+\sqrt{\gamma\left[\sigma_R^2\sum_{i=1}^Rf_{\eta i}^2c_i^2+\sigma_D^2\right]}.
\enn
\eqref{equ:boxnb} is equivalent to $\max_{-\varepsilon_i\leq\eta_i\leq \varepsilon_i} f(\mathbf f_{\eta})\leq 0$. Note that $f(\mathbf f_{\eta})$ is convex in $\mathbf f_{\eta}$ and reaches the maximization at the vertices \cite{ganzhengtsp}. Hence the optimal solution of problem \eqref{equ:q4n} can be obtained by enumerating $2^R$ possibilities of $\mathbf f_{\eta}$, or i.e., $\mathbf f_{\eta}\in \mathcal B$, each one corresponding to an SOCP constraint. Or equivalently
\begin{subequations}\label{equ:box}
\ben
\max _{\mathbf c, \gamma}&& \gamma\\
\text{s.t.} &&\frac{\left(\sum_{i=1}^Rf_{\eta i}c_i\|\mathbf u_i\|_2\right)^2}{\sigma_R^2\sum_{i=1}^Rf_{\eta i}^2c_i^2+\sigma_D^2}\geq \gamma, \mathbf f_{\eta}\in \mathcal B,\\
&& c_i\leq \sqrt{\frac{P_i}{\sigma_R^2+\|\mathbf u_i\|_2^2}}.
\enn
\end{subequations}
Notice that \eqref{equ:box} is equivalent to the form in  \eqref{equ:robb3}, our proof is completed.

\section{Proof of Proposition \ref{pro:monotonic}}\label{app:monotonic}
In this Appendix, we will first prove that problem \eqref{equ:pa2} belongs to the class of monotonic optimization problem, or more specifically, ${\bf SNR}(\mathbf w)$ is an increasing function with respect to $\mathbf w\in \mathcal U$.  Then we will show that problem \eqref{equ:boxnewu} and \eqref{equ:pa2} are equivalent.

In \eqref{equ:worstbu}, we have expressed the worst case SNR as a function of $\mathbf w$, where the power allocation factor $\mathbf c^\sharp$ is adaptively determined as optimal solution of \eqref{equ:robbw} with respect to $\mathbf w$.
For convenience, we further define ${\bf \widetilde{SNR}}(\mathbf c, \mathbf w)$ as a function of $\mathbf w$ and $\mathbf c$, where $\mathbf c$ is only one possible power allocation option rather than the optimal choice, or i.e.,
\begin{subequations}\label{equ:anyc}
\ben
{\bf \widetilde{SNR}}(\mathbf c, \mathbf w)&\triangleq&\min_{\mathbf f_{\eta}\in \mathcal B}\frac{\big(\sum_{i=1}^Rf_{\eta i}c_i\sqrt{w_i}\big)^2}{\sigma_R^2\sum_{i=1}^Rf_{\eta i}^2c_i^2+\sigma_D^2},\label{equ:anyca}\\
\text{s.t.}\quad c_i&\leq&\sqrt{\frac{P_i}{\sigma_R^2+w_i}}\label{equ:anycb}.
\enn
\end{subequations}
Then by definition, one can easily see that ${\bf SNR}(\mathbf w)=\widetilde{{\bf SNR}}(\mathbf c^\sharp, \mathbf w)$.
Suppose $\mathbf w'\geq \mathbf w''$, where $\mathbf w'\triangleq[w_1', \cdots, w'_R]^T$ and $\mathbf w''\triangleq[w_1'', \cdots, w_R'']^T$.  Let $\mathbf c'^{\sharp}\triangleq [c'^{\sharp}_1, \cdots, c'^{\sharp}_R]^T$ and $\mathbf c''^{\sharp}\triangleq [c_1''^{\sharp}, \cdots, c_R''^{\sharp}]^T$ be the optimal solution of \eqref{equ:robbw} for  given $\mathbf w'$ and $\mathbf w''$, respectively.  We will show that ${\bf SNR}(\mathbf w')\geq {\bf SNR}(\mathbf w'')$, or equivalently
\ben\label{equ:toshow}
\widetilde{{\bf SNR}}(\mathbf c'^{\sharp}, \mathbf w')\geq \widetilde{{\bf SNR}}(\mathbf c''^{\sharp}, \mathbf w'').
\enn

Choose one special relay power allocation factor  for the given $\mathbf w'$ as $\tilde {\mathbf c}'\triangleq [\tilde c'_1, \cdots, \tilde c'_R]^T$, such that
\ben\label{equ:optimalc}
\tilde c_i'^{ 2}(w_i'+\sigma_R^2)= c_i''^{\sharp 2}( w_i''+\sigma_R^2), \, i=1, \cdots, R.
\enn
By this condition, the relay powers keep unchanged, and thus the power constraints in \eqref{equ:anycb} are not violated.  Since $w_i'\geq w_i''$, we have $ w_i'+\sigma_R^2\geq  w_i''+\sigma_R^2$. Then $\tilde c^{'2}_i\leq c''^{\sharp 2}_i$ by \eqref{equ:optimalc}, that implies $\tilde {c}'^{2}_i\sigma_R^2\leq c''^{\sharp 2}_i\sigma_R^2$. Then by \eqref{equ:optimalc}, we have
\ben\label{wenadd}
\tilde c'^{2}_iw_i'\geq c''^{\sharp 2}_iw_i^{\sharp}.
\enn

Let $\Gamma_1(\mathbf f_{\eta})\triangleq\frac{\big(\sum_{i=1}^Rf_{\eta i}\tilde c'_i\sqrt{w'_i}\big)^2}{\sigma_R^2\sum_{i=1}^Rf_{\eta i}^2\tilde c'^2_i+\sigma_D^2}$ and $\Gamma_2(\mathbf f_{\eta})\triangleq\frac{\big(\sum_{i=1}^Rf_{\eta i} c''^\sharp_i\sqrt{w''_i}\big)^2}{\sigma_R^2\sum_{i=1}^Rf_{\eta i}^2 c''^{\sharp 2}_i+\sigma_D^2}$ for $\mathbf f_{\eta}\in \mathcal B$. Then we have
\ben
\widetilde{{\bf SNR}}(\tilde {\mathbf c}', \mathbf w')&=&\min_{\mathbf f_{\eta}\in \mathcal B}\Gamma_1(\mathbf f_{\eta}),\label{64wen1}\\
\widetilde{{\bf SNR}}(\mathbf c''^\sharp, \mathbf w'')&=&\min_{\mathbf f_{\eta}\in \mathcal B}\Gamma_2(\mathbf f_{\eta}).\label{64wen2}
\enn

We first
fix some $\mathbf f_{\eta}\in \mathcal B$.
Note that $|\eta_i|\leq \varepsilon_i\leq \|\tilde{\mathbf f}_i\|_2$, we have
$f_{\eta i}=\|\tilde{\mathbf f}_i\|_2+\eta_i\geq 0$.
Then for fixed $\mathbf f_{\eta}$,  the numerator of $\Gamma_1(\mathbf f_\eta)$  is larger than that of $\Gamma_2(\mathbf f_\eta)$ due to \eqref{wenadd}; while the denominator of $\Gamma_1(\mathbf f_\eta)$  is smaller than that of $\Gamma_2(\mathbf f_\eta)$ due to $\tilde c^{'2}_i\leq c^{\sharp 2}_i$.
Hence for any $\mathbf f_{\eta}\in \mathcal B$, we have
\ben
\Gamma_1(\mathbf f_\eta)\geq\Gamma_2(\mathbf f_\eta).\label{wentang}
\enn
%
%
Suppose that the minimum value of $\Gamma_1(\mathbf f_\eta)$ over $\mathbf f_{\eta}\in \mathcal B$ is achieved at $\mathbf f'_\eta$, i. e., $\min_{\mathbf f_{\eta}\in \mathcal B}\Gamma_1(\mathbf f_{\eta})=\Gamma_1(\mathbf f'_\eta)$. Then we have
\ben
\min_{\mathbf f_{\eta}\in \mathcal B}\Gamma_1(\mathbf f_{\eta})=\Gamma_1(\mathbf f'_{\eta})\overset{(a)}\geq \Gamma_2(\mathbf f_{\eta}')\geq \min_{\mathbf f_{\eta}\in \mathcal B}\Gamma_2(\mathbf f_{\eta}),\label{65wen}
\enn
where (a) is due to \eqref{wentang}. Then \eqref{64wen1},\eqref{64wen2} and \eqref{65wen} lead to
\ben\label{equ:snr1}
\widetilde{{\bf SNR}}(\tilde {\mathbf c}', \mathbf w')\geq \widetilde{{\bf SNR}}(\mathbf c''^\sharp, \mathbf w'').
\enn
Since $\tilde {\mathbf c}'$ is just chosen to satisfy \eqref{equ:optimalc}, and may not be  optimal for $\mathbf w=\mathbf w'$, we have
\ben\label{equ:snr2}
\widetilde{ {\bf SNR}}(\mathbf c'^\sharp, \mathbf w')\geq \widetilde{ {\bf SNR}}(\tilde {\mathbf c}', \mathbf w').
\enn
By \eqref{equ:snr1} and \eqref{equ:snr2}, we have \eqref{equ:toshow}, which implies that ${\bf SNR}(\mathbf w)$ is a monotonic increasing function with respect to $\mathbf w$.


On the other hand, $\mathcal U$ has been proved to be convex \cite{RM_Pareto}. Consequently $\mathcal U$ is normal due to the property of convex region~\cite{JB_PA}. Following the similar lines in \cite{RM_Pareto}, it can be shown that $\mathcal U$ is nonempty and closed. Thus \eqref{equ:pa2} is a monotonic optimization problem.

As compared to other nonconvex problems, monotonic problems have the important property that its optimal solution is attained on the Pareto boundary of the feasible region, which can be utilized for solving the problem efficiently.

According to \cite{RM_Pareto}, any Pareto boundary of $\mathcal U$ must be achieved by some rank one matrix $\mathbf G$,  we claim that problem \eqref{equ:boxnewu} and \eqref{equ:pa2} are equivalent.
\section*{Acknowledgement}
We would like to thank the anonymous reviewer for their great constructive comments to improve our work.


\begin{thebibliography}{99}
%

\bibitem{liu}
Y. Liu, and W. Chen, ``Adaptive resource allocation for improved DF aided downlink multi-user OFDM systems,'' {\em IEEE Wireless Commun. Letters},
vol. 1, no. 6, pp. 557-560, Dec., 2012.

\bibitem{FF_Gershman}
H. Chen, S. Shahbazpanahi, and A. B. Gershman, ``Filter-and-forward distributed beamforming for two-way relay networks with frequency selective channels,''
{\em, IEEE Trans. Signal Process.}, vol. 60, no. 4, pp. 1927-1941, Apr., 2012.

\bibitem{Twoway_FF_Schober}
Y. Liang, A. Ikhlef, W. H. Gerstacker, and R. Schober, ``Two-Way Filter-and-Forward Beamforming for Frequency-Selective Channels,'' {\em IEEE Trans. on Wireless Commun.}, vol. 10, no.12, pp. 4172-4183, Dec., 2011.

\bibitem{AF_FF_Schober}
Y. Liang, A. Ikhlef, W. H. Gerstacker, and R. Schober, ''Cooperative Filter-and-Forward Beamforming for Frequency-Selective Channels with Equalization,'' {\em IEEE Trans. on Wireless Commun.}, vol. 10, no. 1, pp. 228-239,  Jan., 2011.

\bibitem{wang}
Z. Wang, W. Chen, and J. Li, ``Efficient beamforming for MIMO relaying broadcast channel with imperfect channel estimation,'' {\em IEEE Trans.
Vehicular Technol.}, vol. 61, no. 1, pp. 419-426, Jan., 2012.

\bibitem{zhang}
Y. Zhang, H. Luo, and W. Chen, ``Efficient relay beamforming design with SIC detection for dual-Hop MIMO relay networks,'' {\em IEEE Trans. Vehicular
Technol.}, vol. 59, no. 8, pp. 4192-4197, Oct., 2010.



%

\bibitem{Power_Jing}
Y. Jing, H. Jafarkhani, ``Network beamforming using relays with perfect channel information,'' {\em IEEE Trans. Inf. Theory} ,  vol. 55, no. 6, pp.
2499-2517, Jun., 2009.

\bibitem{BK_grass}
B. Khoshnevis, W. Y, and R. Adve, ``Grassmannian beamforming for MIMO amplify-and-forward relaying,'' {\em IEEE J. Sel. Area Commun.}, vol. 26, no. 8, pp.
1397-1408, Oct., 2008.

 \bibitem{AF-BF}
Y. Liang, and R. Schober, ``Cooperative amplify-and-forward beamforming with multiple multi-antenna relays,''  {\em IEEE Trans. Commun.}, vol. 59, no. 9,
pp. 2605-2615, Sep., 2011.








\bibitem{PU_Robust_ICC}
P. Ubaidulla, and A. Chockalingam, ``Robust distributed beamforming for wireless relay networks,'' {\em IEEE 20th Int. Sym.}, Sep. 2009, pp.
2345-2349.

\bibitem{ganzhengtsp}
G. Zheng, K. K Wong, A. Paulraj, and B.Ottersten, ``Robust collaborative-relay beamforming,''  {\em IEEE Trans. Signal Process.}, vol. 57, no. 8, pp. 3130-3143, Aug., 2009.

\bibitem{HS_Worst}
H. Shen, W. Xu, J. Wang and C. Zhao, ``A worst case robust beamforming design for multi-antenna AF relaying,''  {\em IEEE Trans Commun. letter}, vol. 17, no. 4, pp. 1089-7798, Apr., 2013.

\bibitem{ZW-wnt}
Z. Wang, and W. Chen, ``Relay Beamforming Design with SIC Detection for MIMO Multi-Relay Networks with Imperfect CSI,'' {\em IEEE Transactions on Vehicular Tech.}, vol. 62, no. 8, pp. 3774-3785, Oct., 2013.

\bibitem{HS_Worst_MR}
H. Shen, J. Wang, B. C. Levy, and C. Zhao, ``Robust optimization for amplify-and-forward MIMO relaying from a worst-case perspective,''  {\em IEEE Trans Signal Process.}, vol. 61, no. 21 pp. 5458-5471, Nov., 2013.

\bibitem{wan}
H. Wan, and W. Chen, ``Joint source and relay design for multi-user MIMO non-regenerative relay networks with direct links,"  {\em IEEE Trans.
Vehicular Technol.}, vol. 61, no.6,  pp. 2871-2876, Jul., 2012.

\bibitem{Tao_2}
R. Wang, M. Tao, and Z. Xiang, ``Nonlinear precoding design for mimo amplify and forward two-way relay systems,'' {\em IEEE Trans. Veh. Technol.}, vol. 61, no. 9,  pp. 3984-3995, Nov., 2012.

\bibitem{Tao_4}
R. Wang, M. Tao, and Y. Huang, ``Linear precoding designs for amplify-and-forward multiuser two-way relay systems,'' {\em IEEE Trans. Wireless
Commun.}, vol. 11, no. 12, pp. 4457-4469, Dec., 2012.

\bibitem{Tao}
M. Tao, and R. Wang, ``Robust relay beamforming for two-way relay networks,'' {\em IEEE Trans Commun. letter}, vol. 16, no. 7, pp. 1052-1055, June, 2012.

\bibitem{Tao_3}
R. Wang, and M. Tao, ``Joint source and relay precoding designs for mimo two-way relaying based on mse criterion,'' {\em IEEE J. Sel. Areas. Commun.},
vol. 60, no. 3,  pp. 1352-1365, Mar., 2012.

\bibitem{Tao_1}
J. Zou, H. Luo, M. Tao, and R. Wang, ``Joint source and relay optimization for non-regenerative mimo two-way relay systems with imperfect CSI,'' {\em
IEEE Trans. Wireless Commun.}, vol. 11, no. 9, pp. 3305-3315, Sep., 2012.

\bibitem{Tao_5}
A. Aziz, Z. Meng, Z. Jianwei, C.N. Georghiades, and C. Shuguang, ``Robust beamforming with channel uncertainty for two-way relay neworks,'' {\em  IEEE
proc. ICC}, Jun. 2012, pp. 3632-3636.


%


%




\bibitem{CK_MulPoint}
C. Kuo, S. Wu, and C. Tseng, ``Robust linear beamfomer desings for coordinated multi-point AF relaying in downlink multi-cell networks,'' {\em IEEE
Trans. Vehicular Technol.}, vol. 11, pp. 3272-3283, Sep., 2012.

\bibitem{nonrob}
B.K. Chalise, and L. Vandendorpe, ``Optimization of MIMO relays for multipoint-to-multipoint communications: nonrobust and robust designs,'' {\em IEEE
Trans. Signal. Process.}, vol. 58, no. 12, pp. 6355-6368, Dec., 2010.

\bibitem{convex}
S. Boyd,  and L. Vandenberghe, ``Convex Optimization,'' Cambridge University Press, 2004.

\bibitem{lectures}
A Nemirovski, ``Lectures on modern convex optimization'', Society for Industrial and Applied Mathematics (SIAM), 2011.

\bibitem{LZ_Realvalue}
L. Zhang, W. Liu, and J. Li, ``Low-complexity distributed beamforming for relay networks with real-valued implementation,'' {\em IEEE
Trans. Signal. Process.}, vol. 61, no. 20, pp. 5039-5048, Oct., 2013.

\bibitem{WW_Dinkelbach}
W. Wang, S. Jin, and Fu. Zheng, ``Maximin SNR beamforming strategies for two-way relay channels,'' {\em IEEE Commun. Letter}, vol. 16, no. 7, pp. 1006-1009, Jul., 2012.

\bibitem{ZF_Dinkelbach}
Z. Fang , X. Wang, and X. Yuan, ``Beamforming design for multiuser two-way relaying: a unified approach via max-min sinr,'', {\em IEEE Trans. Signal Processing}, vol. 61, no. 23, pp. 5841-5852, Dec., 2013.

\bibitem{LL_ploy}
L. Liu, R, Zhang, K. C. Chua. ``Achieving global optimality for weighted sum-rate maximization in the K-user Gaussian interference channel with
multiple antennas,'' {\em IEEE  Trans Wireless Commun.}, vol. 11, no. 5, pp.  1933-1945,  May., 2012.

\bibitem{EB_BRB}
E. Bj\"{o}rnson, G. Zheng, M. Bengtsson, and B. Ottersten, ``Robust monotonic optimization framework for multicell MISO systems,'' {\em IEEE Trans.
Signal Process.}, vol. 60, no. 5, pp. 2508-2523, May, 2012.

\bibitem{EB_BRB2}
E. Bj\"{o}rnson and E. Jorswieck, ¡°Optimal resource allocation in coordinated multi-cell systems,¡± {\em Found. Trends Commun. Inf. Theory}, vol. 9,
no. 2-3, pp. 113-381, 2013.

\bibitem{RM_Pareto}
R. Mochaourab, E. A. Jorswieck, ``Optimal beamforming in interference networks with perfect local channel information,'' {\em IEEE Trans. Signal Process.}, vol. 59, no.3, pp. 1128-1141, Mar., 2011.

\bibitem{CVX}
M. Grant and S. Boyd, CVX' Users' Guide, 2009, [Online], http://cvxr.com/cvx/doc/index.html.

\bibitem{Luo-SDP}
Z.-Q. Luo, W. Kin Ma, A.M.-C. So, Y. Ye, and S. Zhang, ''Semidefinite relaxation of quadratic optimization problems,'' {\em IEEE Signal Process. Mag.}, vol. 27, no. 3, pp. 20-34, 2010.









%
%
%
%
%
%
%
%
%
%
%
%
%
%


%





\bibitem{JB_PA}
J. Brehmer, {\em Utility Maximization in Nonconvex Wireless Systems}. Springer, 2012.

\bibitem{socp_complexity}
M. Lobo, L. Vandenberge, S. Boyd, and H. Lebret, ``Applications of second-order cone programming,'' {\em Linear Algebra and its applications}, vol. 284, pp. 193-228, 1998.
\end{thebibliography}
\end{document}